%% file: colorzdd_arxiv.tex
\documentclass[11pt,letterpaper]{article}
\usepackage[margin=1in]{geometry}
\usepackage{amsthm,amssymb,amsmath}
\usepackage[round]{natbib}
\usepackage{graphicx}
\usepackage{setspace}
\usepackage{subcaption}
\usepackage{url}
\usepackage{rotating}
\usepackage[table]{xcolor}

\usepackage{general}

\title{Solving the Pricing Problem in a Branch-and-Price Algorithm for Graph Coloring using
	Zero-Suppressed Binary Decision Diagrams}
\author{David R.~Morrison, Edward C.~Sewell, Sheldon H.~Jacobson}
\date{}

\newcommand{\var}[1]{\textrm{var}(#1)}
\newcommand{\hi}[1]{\textrm{hi}(#1)}
\newcommand{\lo}[1]{\textrm{lo}(#1)}

\newcommand{\true}{\mathbf{1}}
\newcommand{\false}{\mathbf{0}}
\newcommand{\ZF}{Z_{\mathcal{F}}}
\newcommand{\ZFp}{Z_{\mathcal{F}'}}
\newcommand{\usc}{\string_}

\newtheorem{theorem}{Theorem}

\SetKwFunction{RestrictSet}{RestrictSet}
\SetKwFunction{MakeIndSetZDD}{MakeIndSetZDD}
\SetKwFunction{insert}{insert}
\SetKwFunction{append}{append}
\SetKwFunction{BPZ}{B\&P+ZDD}

\begin{document}
\maketitle

\begin{abstract}
Branch-and-price algorithms combine a branch-and-bound search with an exponentially-sized LP
formulation that must be solved via column generation.  Unfortunately, the standard branching rules
used in branch-and-bound for integer programming interfere with the structure of the column
generation routine; therefore, most such algorithms employ alternate branching rules to circumvent
this difficulty.  This paper shows how a \defn{zero-suppressed binary decision diagram} (ZDD) can be
used to solve the pricing problem in a branch-and-price algorithm for the graph coloring problem,
even in the presence of constraints imposed by branching decisions.  This approach facilitates a
much more direct solution method, and can improve convergence of the column generation subroutine.
\end{abstract}

\doublespacing

\section{Introduction}\label{sec:intro}

Branch-and-price algorithms are of increasing interest in many areas of operations research,
including assignment and scheduling problems \citep{Savelsbergh97, Maenhout10}, vehicle routing
problems \citep{Fukasawa06}, graph coloring \citep{Mehrotra96, Malaguti11}, multicommodity flow
problems \citep{Barnhart00}, and cutting stock problems \citep{Vance98, Pisinger07}, among others.
These algorithms combine a branch-and-bound search together with a tight linear programming
relaxation having an exponential number of variables (such a formulation can be derived, for
example, by the Dantzig-Wolfe decomposition method described in \citealp{Dantzig60}).  This LP
relaxation, called the \defn{master problem}, is used to produce good bounds that are used to prune
suboptimal regions of the search space.  Because the LP relaxation is too large to be stored in
memory, it must be solved via \defn{column generation}.

Let $\mathcal{S}$ be the set of variables for the LP relaxation; each of these variables is
associated with a column of the master problem's constraint matrix (thus, the constraint matrix has
an exponential number of columns).  In the column generation method, a related LP called the
\defn{restricted master problem} (RMP) is built using a (small) subset $\mathcal{S}' \subseteq
\mathcal{S}$ of variables.  The RMP can be solved efficiently by standard linear programming
techniques; however, the solution to the RMP is not necessarily optimal for the master problem.
Therefore, a subroutine called the \defn{pricing problem} must be called to either produce a
variable in $\mathcal{S} \setminus \mathcal{S}'$ that may improve the objective value of the RMP, or
provide a guarantee that no such variable exists.  If an improving variable is found, it is added to
$\mathcal{S}'$, and the RMP is re-optimized.  New columns are iteratively added to $\mathcal{S}'$
until the pricing problem reports that no (potentially) improving variables exist, at which point
column generation is terminated and the solution to the RMP is provably optimal for the master
problem.

In this paper, the pricing problem is assumed to be a weighted binary combinatorial optimization
problem which is characterized by a family of ``valid'' subsets of some universe; in a slight abuse
of notation, solutions to the pricing problem are interchangeably referred to as ``variables'',
``columns'', or ``subsets'', where the meaning is clear from context.  The weights associated with
the pricing problem are usually the optimal dual prices for the current solution to the RMP.  Thus,
the pricing problem returns a new variable with negative reduced cost if one exists; if such a
variable exists, it may improve the value of the RMP \citep{Bertsimas97}.  From this perspective,
the pricing problem is a separation oracle for the dual of the RMP, since new variables for the RMP
correspond to additional constraints in the dual.  

Since the pricing problem is itself often NP-hard, and must be solved exactly, solving it is
typically the most computationally-intensive part of a branch-and-price procedure.  Moreover, when
combined with the standard integer branching scheme used in most branch-and-bound algorithms, the
structure of the pricing problem is destroyed \citep{Barnhart98}.  In such a branching scheme, a
variable $x_i$ with fractional value $\alpha$ is selected at a subproblem in the search tree, and
two children are created with additional bounding constraints $x_i \leq \floor{\alpha}$ and $x_i
\geq \ceil{\alpha}$ (when all variables are binary, this is called \defn{$0-1$ branching}).
However, imposing these constraints changes the structure of the dual problem, which in turn
means that a different separation oracle must be queried at each subproblem in order to generate new
columns.  

In effect, the pricing problem at these subproblems no longer seeks a variable of minimum reduced
cost; it now must produce a variable with minimum reduced cost that respects the current branching
decisions.  This problem, called the \defn{constrained pricing problem}, is much harder than the
regular (or \defn{unconstrained}) pricing problem, and is often related to the $\kth$-shortest-path
problem, which is well-known to be a challenging NP-hard problem \citep{Garey79}. 

Moreover, many branch-and-price formulations have an inherent asymmetry due to the large number of
variables in the formulation.  This asymmetry can lead to extremely lopsided search trees if
standard integer branching techniques are used.  For example, in a problem with many covering
constraints (of the form $\sum_i x_i \geq b$), fixing a variable to zero may not induce much change
in the LP relaxation, but fixing a variable to $1$ may immediately satisfy many constraints.  Thus
long paths in the search tree can exist where many variables are fixed to $0$ but no progress
towards a solution is made.

Therefore, most branch-and-price algorithms employ specialized branching rules or other techniques
to avoid eliminating the pricing problem structure, as well as to maintain a more balanced search
tree.  For example, some branching rules modify the problem structure at each subproblem in the
search tree (e.g., the graph coloring rule of \citealp{Mehrotra96}); others branch on original
(non-reformulated) problem variables, or problem constraints \citep{Vanderbeck11}.  A related scheme
by \citet{Morrison14wide} uses a modified branching scheme called \defn{wide branching}, which does
not wholly eliminate calls to the constrained pricing problem, but restructures the search tree in
an attempt to reduce the number of such calls.

An alternate approach, called \defn{robust branch-and-cut-and-price} (BCP), eliminates calls to the
constrained pricing problem by further modifying the RMP so that branching restrictions can be added
without interfering with the pricing problem structure \citep{deAragao03}.  This approach introduces
additional variables and constraints into the RMP to form a linear program called the \defn{explicit
master}, which has the same objective value as the RMP.  Furthermore, branching decisions made by
the algorithm can be communicated to the pricing problem by imposing constraints on the reduced cost
values in the dual of the explicit master.  This approach has been used successfully in many
variants of the capacitated vehicle routing problem \citep{Pessoa08, Fukasawa06}, as well as related
problems such as the capacitated minimum spanning tree problem \citep{Uchoa08}.

However, no algorithm in the literature has described a way to perform branch-and-price without
using techniques like robust BCP or alternative branching rules, which often come at the expense of
ease of implementation and less-direct (global) solution methods.  Alternate branching rules do not
allow variables to be directly fixed to values, but rely on problem structure to implicitly fix
variables.  Similarly, the robust branch-and-cut-and-price methods require the solution of a larger
LP at each subproblem, and again use implicit methods to fix variables.  The wide branching approach
allows variables to be fixed explicitly, but to obtain good performance, it requires the derivation
of a problem-specific branching rule.

Therefore, the primary contribution of this research is to establish an efficient method for solving
the pricing problem in a branch-and-price algorithm for graph coloring that is directly compatible
with the standard integer branching scheme for graph coloring.  Two algorithmic ideas enable this
result: the first is to use a data structure called a \defn{zero-suppressed binary decision diagram}
(ZDD) to compactly store {\em all} valid solutions to the pricing problem.   A linear-time algorithm
is presented which adds restrictions to the ZDD to prohibit previously-generated columns from being
produced a second time, which allows the constrained pricing problem to be solved at every iteration
of column generation.  The second idea combines the above solution procedure with the cyclic
best-first search (CBFS) strategy to overcome the lopsided search trees that can result when using
standard integer branching with exponentially-sized problems.  Computational results are presented
showing nearly order-of-magnitude improvements in solution time for some instances when using these
two ideas, together with a proof of optimality for several previously unsolved instances.

The remainder of this paper is organized as follows: Section~\ref{sec:zdd} defines the ZDD data
structure and shows how it can be used to solve the pricing problem for the graph coloring problem.
This is done in three parts: first, Section~\ref{sec:defn} shows how ZDDs can be used to solve an
arbitrary unconstrained pricing problem; secondly, Section~\ref{sec:restrict} shows how to modify
this ZDD to solve the constrained pricing problem; finally, Section~\ref{sec:gczdd} shows how to
build a ZDD for the pricing problem in the graph coloring problem, namely, the maximal independent
set problem.  Next, Section~\ref{sec:cbfs} describes the cyclic best-first search strategy and how
it is used to mitigate the effects of lopsided search trees.  In Section~\ref{sec:compres}, the
computational results are given showing the effectiveness of the developed algorithm.  Finally,
Section~\ref{sec:end} outlines several future research directions for this technique.

\section{Zero-Suppressed Binary Decision Diagrams}\label{sec:zdd}

A zero-suppressed binary decision diagram \citep{Minato93} is an extension of the \defn{binary
decision diagram} (BDD) data structure proposed by \citet{Lee59} and \citet{Akers78}.  A BDD is a
directed acyclic graph that compactly encodes a binary function.  Previously, BDDs have been used in
circuit design and verification, as well as a number of formal logic applications \citep{Bryant92}.
More recently, BDDs have been used in a number of different optimization applications:
\citet{Bergman12} explore different variable orderings for BDDs used to characterize the independent
sets of a graph, and \citet{Hadzic08} add weights to the edges of a BDD to perform post-optimality
analysis in a discrete optimization setting.  Finally, \citet{Cire12} and \citet{Bergman12b}
describe how to use BDDs to compute upper and lower bounds to prune subproblems in a
branch-and-bound algorithm.

Despite their success in these related areas, BDDs and ZDDs have not appeared in conjunction with
branch-and-price in the literature before.  \citet{Behle07a} give a method for using BDDs to
enumerate vertices and facets of $0/1$ polyhedra (which can be viewed as solving the pricing problem
for a problem which has been reformulated via Dantzig-Wolfe decomposition), but they do not extend
this result to the branch-and-price setting.  Additionally, \citet{Behle07b} uses BDDs to generate
valid inequalities in a branch-and-cut algorithm to perform row generation instead of column
generation.

The use of decision diagrams together with branch-and-price algorithms can provide substantial
benefits to algorithm performance.  This is because decision diagrams often yield a way to compactly
(in practice) store all the columns even for an exponentially-sized integer program.  Note that
column generation techniques must still be used to solve the RMP, because the columns encoded in the
ZDD cannot be operated on directly by the LP solver.  Nonetheless, since the LP solver has
(implicit) access to all columns, the pricing problem can be solved exactly at every iteration of
column generation, which may improve the convergence of the column generation procedure.  In
contrast, most branch-and-price solvers terminate the pricing problem solver as soon as a column
with ``sufficiently negative'' reduced cost is found, due to the difficulty of solving the pricing
problem.  Moreover, as shown in Section~\ref{sec:restrict}, the set of valid pricing problem
solutions can be modified in place, allowing branch-and-price algorithms using ZDDs to employ
standard integer branching methods.

\subsection{Definitions and Notation}\label{sec:defn}

A ZDD is a modified version of a BDD that removes some nodes from the data structure to reduce its
size.  ZDDs are most useful when the binary function it encodes is ``sparse'' in the sense that there
are relatively few valid solutions to the function compared to the number of invalid solutions.
\citet{Minato93} observed that many combinatorial optimization problems have the sparsity
characteristic; thus, ZDDs are likely to be more useful in a branch-and-price setting than ordinary
BDDs.

Formally, a ZDD $Z$ is defined as follows.  Let $\mathcal{E}$ be an ordered set of $n$ elements
$(e_1,e_2,\hdots,e_n)$; then $Z$ is a directed acyclic graph satisfying the following properties:

\begin{enumerate}
\item There are two special nodes in $Z$ (denoted $\true$ and $\false$), called the \defn{true} 
	node and \defn{false} node, respectively.  Additionally, there is exactly one ``highest'' node
	in the topological ordering of $Z$, called the \defn{root} of $Z$, and denoted $r$.
\item Every node $a \in Z - \{\true, \false\}$ has two outgoing edges, a \defn{high edge} and a
	\defn{low edge}, which point to the \defn{high child} and \defn{low child}, respectively.  The
	high (low) child of $a$ is denoted $\hi{a}$ ($\lo{a}$).  The true and false nodes have no
	outgoing edges.  The \defn{indegree} of $a$, denoted $\delta^-(a)$, is the number of incoming
	edges to $a$; thus, $\delta^-(r) = 0$.
\item Every node $a \in Z - \{\true, \false\}$ is associated with some element $e_i \in
	\mathcal{E}$; the index of the associated element for $a$ is given by $\var{a}$, that is,
	$\var{a} = i$.  By convention, $\var{\true} = \var{\false} = n + 1$.  Finally, if $\var{a} = i$,
	then $\var{\hi{a}} > i$ and $\var{\lo{a}} > i$.  
\item No $a \in Z$ has $\hi{a} = \false$ (this property, called the \defn{zero-suppressed} property,
	is not satisfied by ordinary BDDs).
\end{enumerate}

Any set $A \subseteq \mathcal{E}$ induces a path $P_A$ from the root of $Z$ to either $\true$ or
$\false$, in the following manner: starting at the root of $Z$, if $a$ is the current node on the
path, the next node along the path is $\hi{a}$ if $e_{\var{a}} \in A$, and $\lo{a}$ otherwise.  The
\defn{output} of $Z$ on $A$, denoted $Z(A)$, is the last node along this path, which must be either
$\true$ or $\false$.  If $Z(A) = \true (\false)$, then $Z$ \defn{accepts} (\defn{rejects}) $A$.
Note that it is not required for $\var{b} = \var{a} + 1$ when $b$ is a child of $a$; in the case
when $\var{b} > \var{a} + 1$, the edge is called a \defn{long edge}, and when an induced path $P_A$
includes such an edge, if $\{e_{\var{a}+1},e_{\var{a}+2},\hdots,e_{\var{b}-1}\} \cap A \neq \emptyset$,
then $Z$ rejects $A$.  Finally, a ZDD \defn{characterizes} a family of sets $\mathcal{F} \subseteq
2^\mathcal{E}$ (denoted $\ZF$) if $Z$ accepts all sets in $\mathcal{F}$, and rejects all sets not in
$\mathcal{F}$ (see Figure~\ref{fig:zdd}).

\begin{fig}
\caption{\label{fig:zdd} Let $\mathcal{E} = (e_1, e_2, e_3, e_4)$, and $\mathcal{F} = \{\emptyset,
	\{e_1,e_2\}, \{e_3,e_4\}, \{e_1,e_2,e_3,e_4\}\}$ \citep{Andersen97}.  Solid lines represent high
	edges, and dashed lines represent low edges; all edges are directed downwards.  Grey nodes
	indicate whether $\ZF$ accepts $A$.} 
\begin{subfigure}[t]{0.2\textwidth}
	\centering
    \includegraphics[scale=0.5]{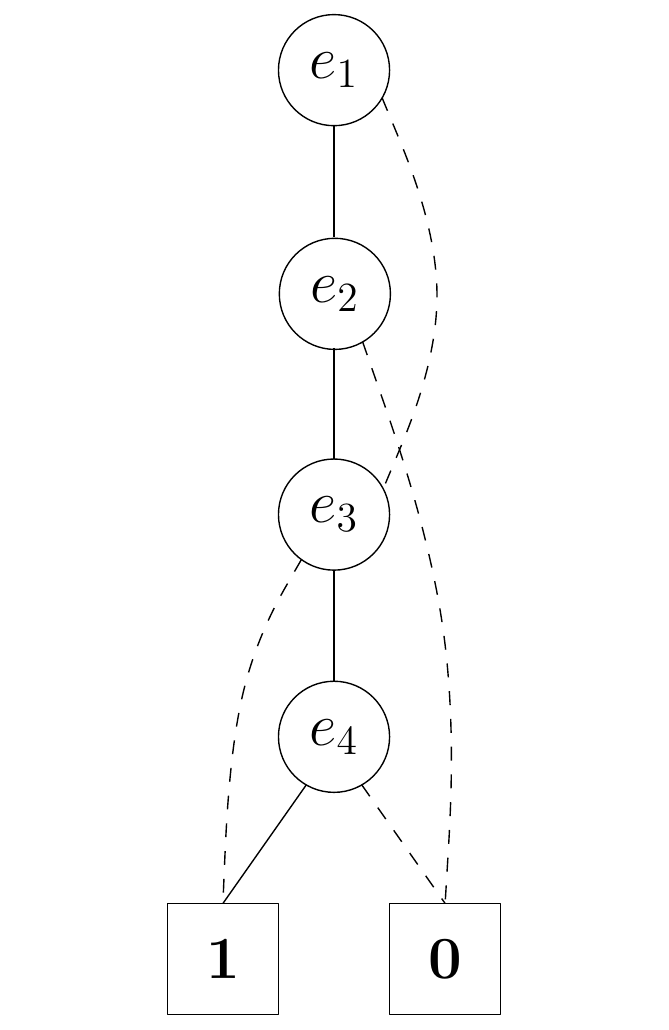}
	\caption{The unique smallest ZDD characterizing $\mathcal{F}$ for the given variable ordering.}
	\label{fig:zdd1}
\end{subfigure}
\qquad
\begin{subfigure}[t]{0.2\textwidth}
	\centering
    \includegraphics[scale=0.5]{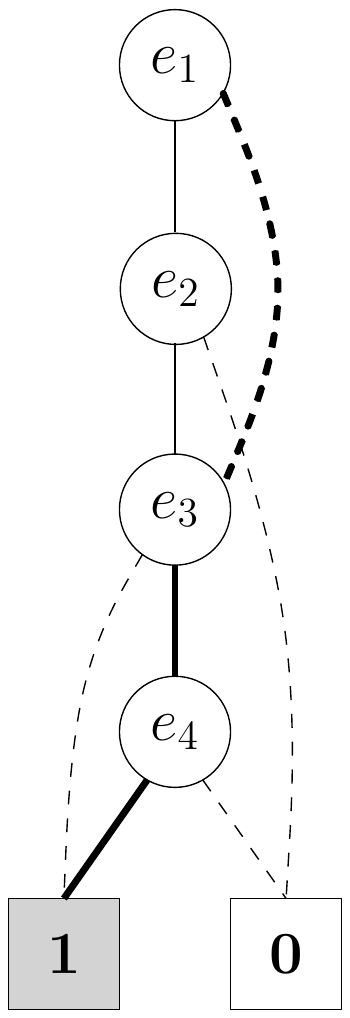}
	\caption{The induced path corresponding to the set $A = \{e_3, e_4\}$.}
	\label{fig:zdd2}
\end{subfigure}
\qquad
\begin{subfigure}[t]{0.2\textwidth}
	\centering
    \includegraphics[scale=0.5]{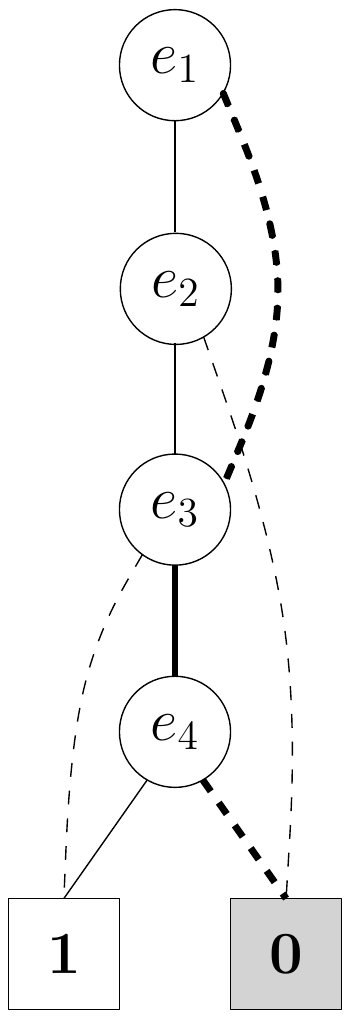}
	\caption{The induced path corresponding to the set $A = \{e_3\}$.}
	\label{fig:zdd3}
\end{subfigure}
\qquad
\begin{subfigure}[t]{0.2\textwidth}
	\centering
    \includegraphics[scale=0.5]{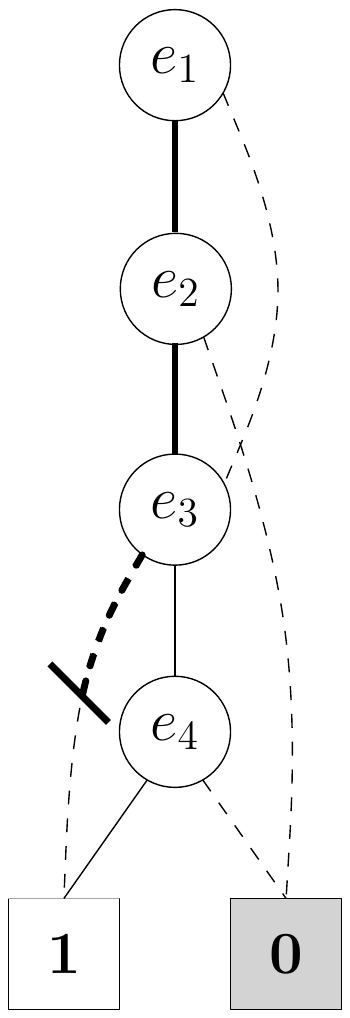}
	\caption{The ZDD does not accept $A = \{e_1,e_2,e_4\}$ since the long edge skips $e_4$, but $e_4
		\in A$.}
	\label{fig:zdd4}
\end{subfigure}
\end{fig}

For an arbitrary family $\mathcal{F}$ and an arbitrary vertex ordering, the size of $\ZF$ (that is,
the number of nodes and edges in the graph, denoted $|\ZF|$) may be exponential in $n$.  However,
\citet{Bryant86} shows that for any fixed variable ordering, every boolean function has a unique
smallest BDD characterizing it.  This result extends to ZDDs by observing that membership in
$\mathcal{F}$ can be defined as a boolean function.  One way to construct the unique smallest ZDD
characterizing $\mathcal{F}$ is to first construct the BDD for $\mathcal{F}$'s indicator function,
and then iteratively delete nodes whose high edge points to $\false$, connecting the low edge to the
node's parent.  Alternately, there exists a recursive algorithm to construct $\ZF$ directly
\citep{Knuth08}.

Note that the choice of ordering on the elements of $\mathcal{E}$ is important; \citet{Bryant86}
shows examples where different variable orderings yield BDDs of dramatically different sizes for the
same function.  In fact, it is NP-hard to determine the variable ordering for any arbitrary boolean
function that will yield the smallest BDD \citep{Bollig96}.  These results apply for ZDDs as well;
nevertheless, the use of heuristic variable orderings often results in tractably-sized ZDDs in
practical applications.

To see how ZDDs can be used to solve the unconstrained pricing problem in a branch-and-price
algorithm, let $\mathcal{F}$ be the set of all valid solutions to the pricing problem.  Then, using
the technique of \citet{Hadzic08}, assign weights to the edges of $\ZF$ and compute the longest path
or shortest path in $\ZF$ from the root to $\true$, depending on whether the pricing problem is a
maximization or minimization problem.  Specifically, let $(\pi_1,\pi_2,\hdots,\pi_n)$ be a weight
vector for the elements of $\mathcal{E}$; set the weight of edge $(a,b) \in \ZF$ to $\pi_{\var{a}}$
if $b = \hi{a}$, and $0$ otherwise.  Then, finding the longest or shortest path with respect to
$\{\pi\}$ from the root of $\ZF$ to $\true$ can be found in $O(|\ZF|)$ time using dynamic
programming \citep{Sedgewick11}.  The resulting path corresponds to the optimal solution to the
pricing problem (see Figure~\ref{fig:zdd_opt}).

\begin{fig}
\includegraphics[scale=0.7]{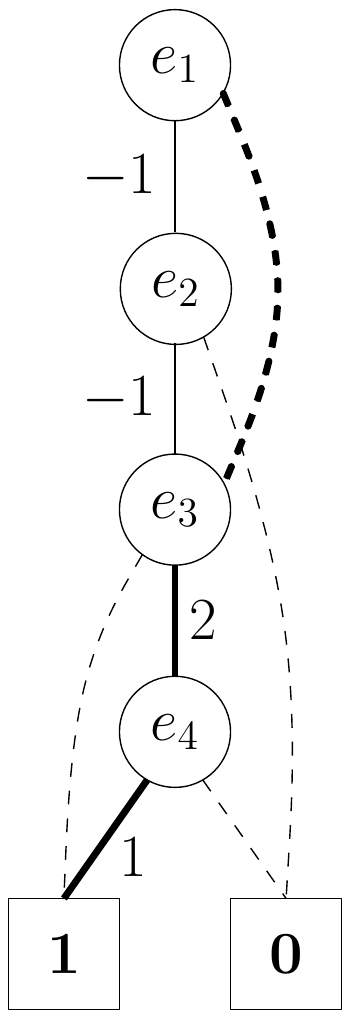}
\caption{\label{fig:zdd_opt} The ZDD from Figure~\ref{fig:zdd1} with weights given by the objective
	function $\max \brac{-e_1 - e_2 + 2e_3 + e_4}$; the bold path corresponds to the maximum-weight
	valid set, that is $\{e_3, e_4\}$.  Weights not shown are $0$.}
\end{fig}

\subsection{The ZDD Restriction Algorithm}\label{sec:restrict}

In order to use standard integer branching methods in a branch-and-price algorithm it is necessary
to solve the constrained pricing problem.  Recall that this problem seeks a new variable of minimum
reduced cost that respects all branching decisions made at the subproblem.  Note that it is
sufficient to generate a new variable that does not appear in the pool $\mathcal{S}'$ for the RMP;
to see this, observe that if any variable in $\mathcal{S}'$ has negative reduced cost, then the
current solution to the RMP is not optimal.  Therefore, in this section, an algorithm is presented
to add restrictions to a ZDD characterizing the pricing problem so that any time a new column is
generated and added to $\mathcal{S}'$, it can be immediately restricted from ever being generated as
a solution to the pricing problem again.  In this way, the ZDD will actually solve the constrained
pricing problem at each iteration of the algorithm.

Let $\mathcal{F}$ be the family of valid solutions to the pricing problem, where each $A \in
\mathcal{F}$ is a subset of $\mathcal{E} = (e_1,e_2,\hdots,e_n)$, and let $\ZF$ be the ZDD
characterizing $\mathcal{F}$.  The restriction algorithm for ZDDs, called \RestrictSet, takes as
input a set $A \in \mathcal{F}$, and builds a new ZDD $\ZFp$ that accepts $\mathcal{F}' =
\mathcal{F} - A$.  The key feature of the ZDD restriction algorithm that makes it effective in
practice is that it operates in $O(n)$ time, and it increases the size of $\ZF$ by at most $n$ nodes
and $2n$ edges (and often by much less).  

Intuitively, the \RestrictSet algorithm identifies the path $P_A$ in $\ZF$ corresponding to the set
$A$, and updates this path so that it ends at the false node instead of the true node.  However, if
there exists $A' \neq A$ such that $P_A$ and $P_{A'}$ overlap, this update could also restrict $A'$.
Therefore, \RestrictSet duplicates the portion of $P_A$ that could overlap with some other
root-to-$\true$ path, and sets the endpoint of the duplicate path to $\false$.  This ensures that no
additional sets are restricted by the algorithm.  The first node on $P_A$ with indegree greater than
one, referred to as the \defn{split node}, is the first node with some potential overlap; thus it,
and all subsequent nodes, are duplicated.

Pseudocode for the \RestrictSet algorithm is given in Algorithm~\ref{alg:rset}; this algorithm makes
use of a function called $\ZF.\insert(i,a_1,a_2)$, which takes as input an index $i \in
\{1,2,\hdots,n\}$ and pointers to two pre-existing nodes $a_1,a_2 \in \ZF$.  The function inserts a
new node into $\ZF$ associated with element $e_i$, with low edge pointing to $a_1$ and high edge
pointing to $a_2$, and returns a pointer to the newly-inserted node.  It also updates the indegrees
of the high and low children.  $\ZF.\insert$ can be implemented in (average) constant time (see
\citealp{Andersen97} for details).

\begin{alg}[t]
    \caption{\label{alg:rset} RestrictSet($\ZF, A$)}
    \Input{A ZDD $\ZF$ and a characteristic vector $(\alpha_1,\alpha_2,\hdots,\alpha_n)$ for a set $A
		\in \mathcal{F}$}
    \Output{A modified ZDD $\ZFp$ such that $\mathcal{F}' = \mathcal{F} - A$}

    \BlankLine
    \Comment{Find the first node on $P_A$ with indegree higher than $1$}
    $a = \textrm{root}(\ZF)$; $b = -1$\;
    \While{$\delta^-(a) < 2$ and $a \nin \{\true, \false\}$}
    {
        $i = \var{a}$; $b = a$\;
        \lIf{$\alpha_i$} { $a = \hi{a}$\; }
        \lElse{ $a = \lo{a}$\; }
    }
    \Comment{The node $c$ is the ``split'' node, and $b$ is its parent}
    $c = a$\;

    \BlankLine
    \Comment{Make copies of all remaining nodes on $P_A$ and point to $\false$}
    $\textrm{list} = ()$\;
    \While{$a \nin \{\true, \false\}$}
    {
        $\textrm{list}.\append(a)$\;
        \lIf{$\alpha_{\var{a}}$} { $a = \hi{a}$\; }
        \lElse{ $a = \lo{a}$\; }
    }

    \BlankLine
    $a' = \false$ \label{rset:zero}\;
	\Comment{Insert the duplicated nodes into $\ZF$}
    \ForEach{$a \in \textrm{list}$ (in reverse order)}
    {
        \lIf{$\alpha_{\var{a}}$} {$a' =\ZF.\insert(\var{a},\lo{a},a')$\label{rset:b_add}\;}
        \lElse { $a' = \ZF.\insert(\var{a}, a', \hi{a})$\label{rset:e_add}\;}
    }

	\BlankLine
	\Comment{Point the correct edge of the parent node $b$ to the root of the duplicated path}
    \lIf{$\alpha_{\var{b}}$}{ $\hi{b} = a'$\label{rset:par1}\; }
    \lElse{ $\lo{b} = a'$\label{rset:par2}\; }
	return $\ZF$\;
\end{alg}

The following theorem establishes the correctness of the \RestrictSet algorithm and proves the
claims made previously about its time and space complexity behavior; an example of the \RestrictSet
applied to the ZDD in Figure~\ref{fig:zdd1} is given in Figure~\ref{fig:zddrestr}.

\begin{theorem}\label{thm:rset}
	Given a ZDD $\ZF$ describing a family of subsets $\mathcal{F}$ of an ordered set $\mathcal{E}$
	with $n$ elements, together with a set $A \in \mathcal{F}$, the \RestrictSet algorithm modifies
	$\ZF$ in $O(n)$ time to produce a new ZDD called $\ZFp$ such that $\mathcal{F}' = \mathcal{F} -
	A$.  Furthermore, $|\ZFp| \leq |\ZF| + 3n$.
\end{theorem}

\begin{proof}
	First, note that \RestrictSet visits each node along $P_A$ at most twice, and $P_A$ has at most
	$n$ nodes.  Furthermore, the algorithm performs a constant amount of work for each visited node.
	Thus the running time of \RestrictSet is $O(n)$.  Also, since node $c$ is at most the root of
	$\ZF$, at most $|\mathcal{E}|$ nodes are added to $\ZF$ to form $\ZFp$, and each new node has
	two outgoing edges.

	To prove that $\ZFp$ has the desired properties, let $a_1',a_2',\hdots,a_l'$ be the nodes added
	to $\ZF$ in lines \ref{rset:b_add}-\ref{rset:e_add} (Algorithm~\ref{alg:rset}), in increasing
	order of depth.  Consider some set $A' \subseteq \mathcal{E}$; if $A' = A$, the path from the
	root of $\ZFp$ to the bottom of the ZDD is the same as the path from the root of $\ZF$ up to the
	parent $b$ of the split node $c$.  By construction, the next node visited in $\ZFp$ is $a_1'$
	(lines \ref{rset:par1} and \ref{rset:par2}, Algorithm~\ref{alg:rset}).  Then, the remainder of
	the path in $\ZFp$ follows the added nodes; at each $a_i'$, the high and low children are
	constructed to agree with the values of $A$.  Finally, the last node along this path is $\false$
	(line \ref{rset:zero}, Algorithm~\ref{alg:rset}), so $\ZFp(A') = 0$.

    Furthermore, if $A' \neq A$, then consider the first index $i$ where the characteristic vectors
	of $A$ and $A'$ differ; if $i < \var{c}$, then the modifications to $\ZFp$ have no effect on
	whether $A'$ is accepted, since the only nodes added to $\ZFp$ appear at depths greater than or
	equal to that of $c$.  However, if $i \geq \var{c}$, the path will follow along the newly added
	nodes $a_1',a_2',\hdots,a_j'$, where $\var{a_j'} = i$.  At this point, $A$ and $A'$ differ, and by
	construction, the path returns to the original node in $\ZF$ and never returns to the
	newly-added nodes.  Therefore, $\ZFp(A') = \ZF(A')$, as desired.
\end{proof}

To reduce the size of $\ZFp$, a check can be performed to see if the high and low edges of
newly-inserted nodes both point to $\false$; in this case, the node is suppressed (see
Figure~\ref{fig:zdd_res3}).  Finally, note that the ZDD produced by the \RestrictSet is no longer
necessarily minimal with respect to $\mathcal{F}'$.  In particular, in the worst case, if all $2^n$
subsets of $\mathcal{E}$ are restricted, the size of the ZDD can grow by $O(n2^n)$ nodes.  However,
in this case, the resultant ZDD is $Z_\emptyset$, which can be described with only two nodes.  In
the event that the ZDD becomes too large, a reduction algorithm can be periodically called that
searches for duplicate nodes in $\ZF$ that can be merged.

\begin{fig}
\caption{\label{fig:zddrestr} The result of applying the \texttt{RestrictSet} algorithm to $\ZF$
	from Figure~\ref{fig:zdd1} with $A = \{e_1,e_2,e_3,e_4\}$.  The final ZDD accepts $\mathcal{F}'
	= \{\emptyset, \{e_1,e_2\}, \{e_3,e_4\}\}$.}
\begin{subfigure}[t]{0.28\textwidth}
	\centering
    \includegraphics[width=\textwidth]{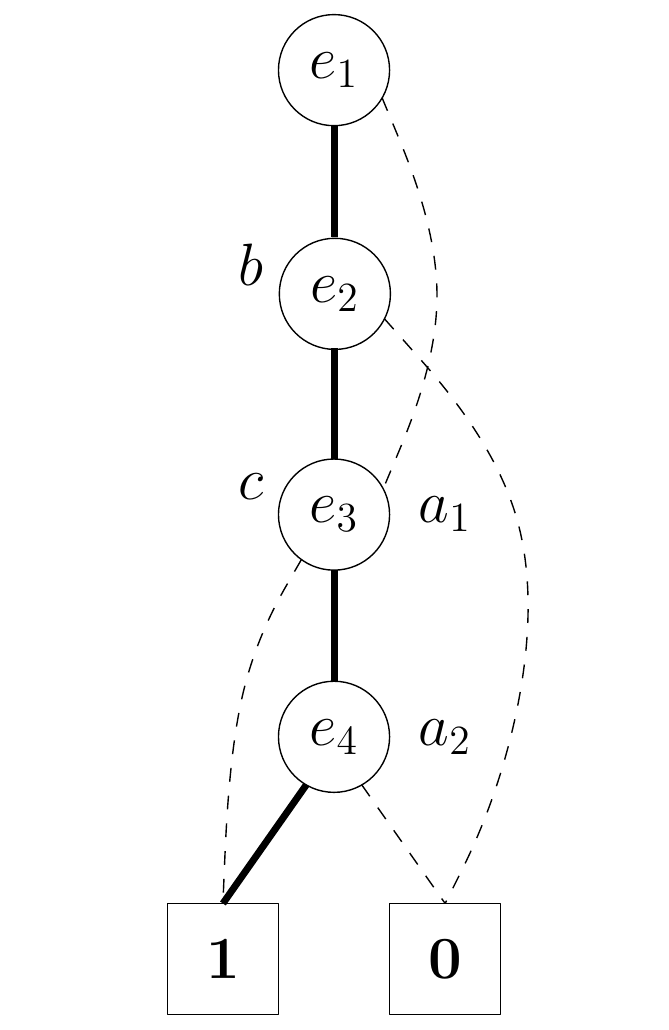}
	\caption{The path $P_A$ is in bold; the split node $c$ is the first node along
		this path with indegree larger than $1$.  The parent of the split node is $b$.}
	\label{fig:zdd_res1}
\end{subfigure}
\qquad
\begin{subfigure}[t]{0.28\textwidth}
	\centering
    \includegraphics[width=\textwidth]{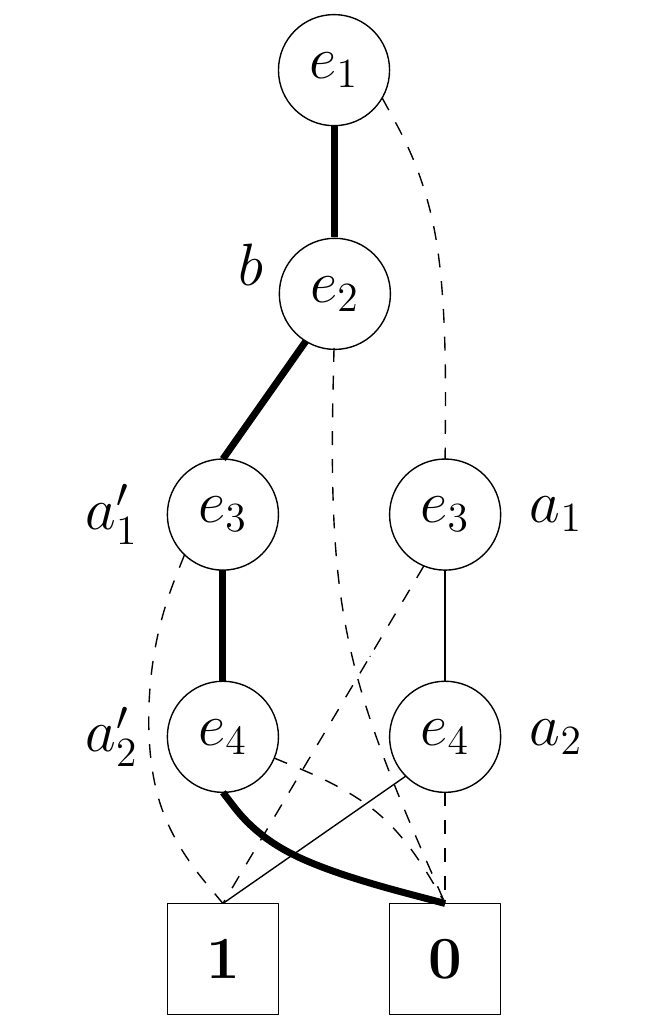}
	\caption{Copies of nodes $a_1$ and $a_2$ are created, and the high edge from $b$ points to this
		new path.  The new path points to $\false$, thus restricting the set $A$.}
	\label{fig:zdd_res2}
\end{subfigure}
\qquad
\begin{subfigure}[t]{0.28\textwidth}
	\centering
    \includegraphics[width=\textwidth]{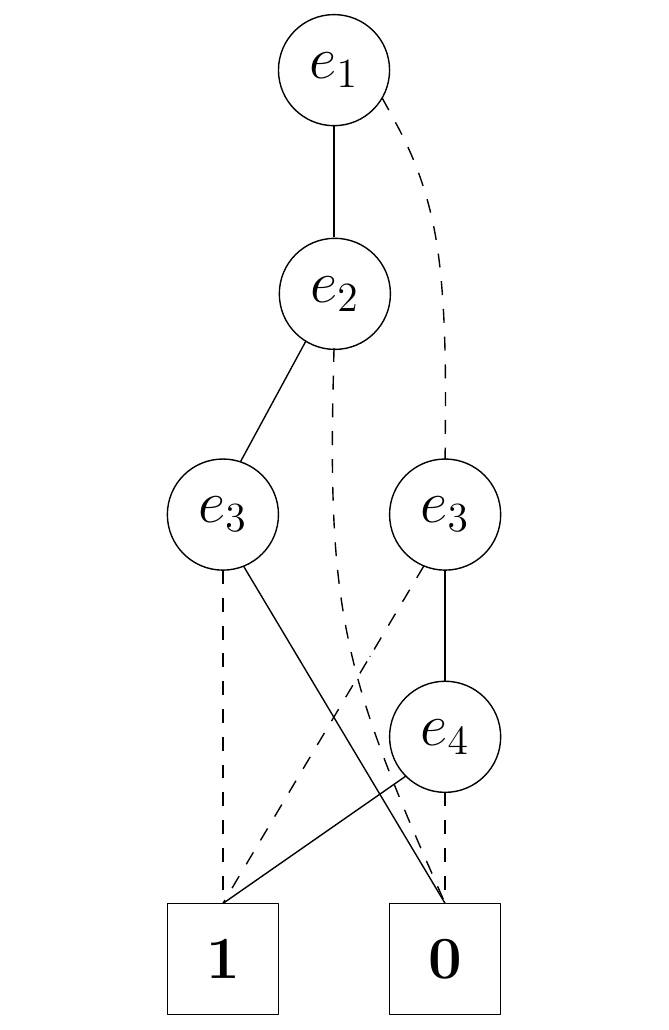}
	\caption{The new ZDD $Z_{\mathcal{F} - A}$; since both high and low edges of $a_2'$ point to
		$\false$, it can be suppressed.  $a_1'$ is also suppressed to satisfy the zero-suppressed
		property.}
	\label{fig:zdd_res3}
\end{subfigure}
\end{fig}

Using the \RestrictSet procedure, a branch-and-price algorithm can be developed that uses
traditional integer branching.  This branch-and-price algorithm first builds a ZDD characterizing
all valid solutions to the pricing problem; in the worst case, this may take exponential time, but
dynamic programming or memoization techniques can be used to speed up the construction.  The ZDD is
then used to produce new variables at every iteration of column generation, which correspond to
solutions of the constrained pricing problem.  Once a new set (or variable) has been generated,
\RestrictSet is called to prohibit that column from being generated again at a later time.  The ZDD
is therefore guaranteed to produce the optimal solution to the pricing problem at each stage, and
since in most cases $n \ll |\ZF|$, the increase in size of the ZDD over the course of the
branch-and-price search is small.  Hence, the time needed to solve the pricing problem does not
significantly increase over the course of the algorithm.  Pseudocode for the resulting
branch-and-price search is given in Algorithm~\ref{alg:bpzdd}.

\begin{alg}
\caption{\label{alg:bpzdd} Branch-and-Price with ZDDs}
Construct $\ZF$, where $\mathcal{F}$ is the set of valid columns\;
Compute an initial pool $\mathcal{S}'$ of columns for the RMP\;
\lForEach{$A \in \mathcal{S}'$}{$\ZF = \RestrictSet(\ZF, A)$\;}
Initialize the branch-and-price search tree $T$ \label{bpzdd:init}\;
\BlankLine
\Comment{Main branch-and-price loop}
\While{$T$ has an unexplored subproblem}
{
	Select a subproblem $s$ that has not been explored \label{bpzdd:select}\;
	Generate children of $s$ according to branching rule\;
	\ForEach{child of $s$}
	{
		\BlankLine
		\Comment{Column generation loop}
		\While{$\exists$ a variable in $\mathcal{S} \setminus \mathcal{S}'$ with negative 
			reduced cost}
		{
			Use $\ZF$ to generate a variable $A \in \mathcal{S} \setminus \mathcal{S}'$ 
				with negative reduced cost\;
			$\RestrictSet(\ZF, A)$\;
			Add $A$ to $\mathcal{S}'$ and re-optimize the RMP\;
		}
		\BlankLine
		Apply pruning rules to delete child, or insert child into $T$\;
	}
	Mark $s$ as explored\;
}
return the best solution found in $T$\;
\end{alg}

\subsection{The Maximal Independent Set ZDD}\label{sec:gczdd}

The graph coloring problem is a classic NP-hard problem \citep{Garey79}; given a graph $G=(V,E)$,
the objective is to find a minimum \defn{proper coloring} of vertices (i.e., a coloring in which no
adjacent vertices share colors).  The \defn{chromatic number} $\chi$ of $G$ is the minimum number of
colors required in any proper coloring.  For a vertex $v$, the \defn{neighborhood} of $v$, denoted
by $N(v)$, is the set of vertices adjacent to $v$.  For a subset $S \subseteq V$, the \defn{induced
subgraph} $G[S]$ is the subgraph of $G$ with vertex set $S$ that has an edge between $u,v \in S$ if
and only if $(u,v)$ is an edge in $G$.  A set $S \subseteq V$ is an \defn{independent set} if $G[S]$
has no edges, and $S$ is a \defn{maximal} independent set if there is no vertex $v \in V \setminus
S$ such that $S + v$ is independent.  For any set of vertices $S$, a vertex $v$ is \defn{covered} if
$v \in S$ or $v \in N(S)$.  Finally, a set $S \subseteq V$ is a \defn{clique} if all pairs of
vertices in $S$ are adjacent.  

The integer programming formulation for graph coloring used in most state-of-the-art solvers, 
initially proposed by \citet{Mehrotra96}, is as follows: 

\begin{equation}\label{eqn:ip2}
\begin{aligned}
\textrm{minimize } & \sum_{S \in \mathcal{S}} x_S\\
\textrm{subject to } & \sum_{S : v \in S} x_S \geq 1 \forall v \in V\\
& x_S \in \{0,1\}.
\end{aligned}
\end{equation}

\noindent
In this formulation, $\mathcal{S}$ is the (exponential) family of maximal independent sets in $G$;
since any proper coloring can be viewed as a partition of $V$ into independent sets, this is
equivalent to searching for the smallest coloring.  The binary variables $x_S$ indicate whether the
maximal independent set $S$ is used in the coloring, and the constraints ensure that each vertex in
the graph appears in some color class.  

The pricing problem for the graph coloring problem as formulated in \pref{eqn:ip2} is a
maximum-weight maximal independent set problem, where the weights on the vertices are given by the
values of the dual variables of the RMP.  If a maximal independent set $S$ with weight larger than
$1$ is found, then variable $x_S$ has negative reduced cost, which means that $x_S$ is a candidate
to improve the solution value of the RMP and can be added to $\mathcal{S}'$.

These solvers often use a branching rule called \defn{edge branching} to avoid destruction of the
pricing problem.  Edge branching selects two non-adjacent, uncolored vertices in $G$ and creates two
branches, one in which the vertices are linked by an edge, and one in which they are merged together
\citep{Mehrotra96}.  However, in this paper, a ZDD is built characterizing the family of maximal
independent sets of $G$, and is used with a standard $0-1$ branching method (called \defn{variable
branching} by \citealp{Malaguti11}).

To build a ZDD for the pricing problem of formulation \pref{eqn:ip2}, fix an ordering
$\{v_1,v_2,\hdots,v_n\}$ on the vertices of the graph, and for some vertex $u \in V$, let $u - 1$ be
the vertex immediately preceding $u$ in this ordering (in this setting, the vertex set $V$ plays the
role of $\mathcal{E}$).  This section describes a recursive algorithm for building the minimal ZDD
(with respect to this vertex ordering) that characterizes $\mathcal{F} = \{A \st \textrm{$A$ is a
maximal independent set in $G$}\}$.  

The maximal independent set ZDD is stored as a lookup table; if $a$ is an index in this table, the
lookup table stores $\var{a}$, $\lo{a}$, $\hi{a}$, and $\delta^-(a)$.  In addition, to facilitate
the merging of isomorphic regions of the ZDD, a reverse lookup table is stored that maps a tuple
$(i, b, c)$ to an index $a$ such that $\var{a} = i$, $\lo{a} = b$, and $\hi{a} = c$, if such an $a$
exists in $\ZF$.  This reverse table is implemented as a hash table which allows for average
constant insertion and lookup time.  

A recursive construction algorithm for $\ZF$, called \MakeIndSetZDD, can be formulated following the
general approach given in \citet{Knuth08}.  At each stage, a set $U$ of $k$ vertices and an index $i
\leq k+1$ is given as input to \MakeIndSetZDD; $U$ is the set of uncovered vertices for some (not
necessarily maximal) independent set $R \subseteq \{v_1,v_2,\hdots,u_i - 1\}$.  Here, vertices
$\{u_i,u_{i+i},\hdots,u_k\}$ can still be added to this (hypothetical) set $R$.  To construct the ZDD
node corresponding to $U$ and $i$, the high child $b_h$ and low child $b_\ell$ must first be
constructed.  To do this, vertex $u_i$ and all its neighbors are removed from $U$ to form a set
$U_H$, and $h$ is set to the index (in $U_H$) of the first vertex appearing after $u_i$ in the
vertex ordering, or $|U_H| + 1$, if no such vertex exists.  Then, to compute $b_h$,
$\MakeIndSetZDD(U_H, h)$ is called; this mimics the addition of vertex $u_i$ to an independent set
$R$ at the current node.  Conversely, to compute the low child, $\MakeIndSetZDD(U,i+1)$ is called,
which forbids vertex $u_i$ from being used in $R$.  To ensure minimality of $\ZF$, before a node
corresponding to some set $U$ and index $i$ is inserted, the algorithm checks to see if any node $a$
exists in the ZDD with $\var{a} = i$, $\hi{a} = b_h$, and $\lo{a} = b_\ell$.  If such a node exists,
the index of that node is returned; otherwise a new node is inserted.  

In the base case, $i = k+1$ (that is, no more vertices can be added to an independent set from the
current node).  If $U$ is empty, all vertices in $G$ are covered by an independent set, so the
algorithm returns $\true$; if $U$ is not empty, there is some uncovered vertex in $U$, so the
algorithm returns $\false$ (in fact, a tighter base case can be developed by observing that if there
is some vertex in $\{u_1,u_2,\hdots,u_{i-1}\}$ that is not adjacent to any vertex in
$\{u_i,u_{i+1},\hdots,u_k\}$, it is impossible to build a maximal independent set from the current ZDD
node).  

Pseudocode for \MakeIndSetZDD is given in Algorithm~\ref{alg:mkzdd}; to build $\ZF$,
$\MakeIndSetZDD(V, v_1)$ is called.  An example of running \MakeIndSetZDD on the graph in
Figure~\ref{fig:ex} is given in Figure~\ref{fig:steps}.

\begin{alg}
\caption{\label{alg:mkzdd} MakeIndSetZDD$(U, i)$}
\Input{A set $U = \{u_1,u_2,\hdots,u_k\}$ of uncovered vertices such that $u_j < u_{j+1}$ with respect
    to the vertex ordering on $V$, and a ``current index'' $i$ }
\Output{The root node of a ZDD characterizing all the maximal independent sets in $G[U]$ that can be
    formed with vertices in $\{u_i,u_{i+1},\hdots,u_k\}$}
\lIf{$G[U]$ cannot be covered by taking all vertices in $\{u_i,u_{i+1},\hdots,u_k\}$}
    {return $\false$\label{rec:base1}\;}
\lIf{$U == \emptyset$}{return $\true$\label{rec:base2}\;}
\BlankLine
$U_H = U - u_i - N(u_i)$ \Comment{Use vertex $u_i$; remove it and its neighbors from $U$}
$h = \min\{j \st u_j > u_i \textrm{ and } u_j \in U_H\}$ or $|U_H|+1$ if no such $j$ exists\;
$b_h = \MakeIndSetZDD(U_H, h)$\;
$b_\ell = \MakeIndSetZDD(U, i+1)$\;
\BlankLine
\lIf{$b_h == \false$}{return $b_l$\label{rec:zero}\;}
\Comment{Use reverse lookup table}
\lIf{$\exists a \in \ZF \textrm{ s.t. } \var{a} = i, \lo{a} = b_\ell, \textrm{ and } \hi{a} = b_h$}
    {return $a$\label{rec:hash}\;}
\lElse{return $\ZF.\texttt{insert}(i, b_\ell, b_h)$\;}
\end{alg}

\begin{fig}
\includegraphics[scale=0.5]{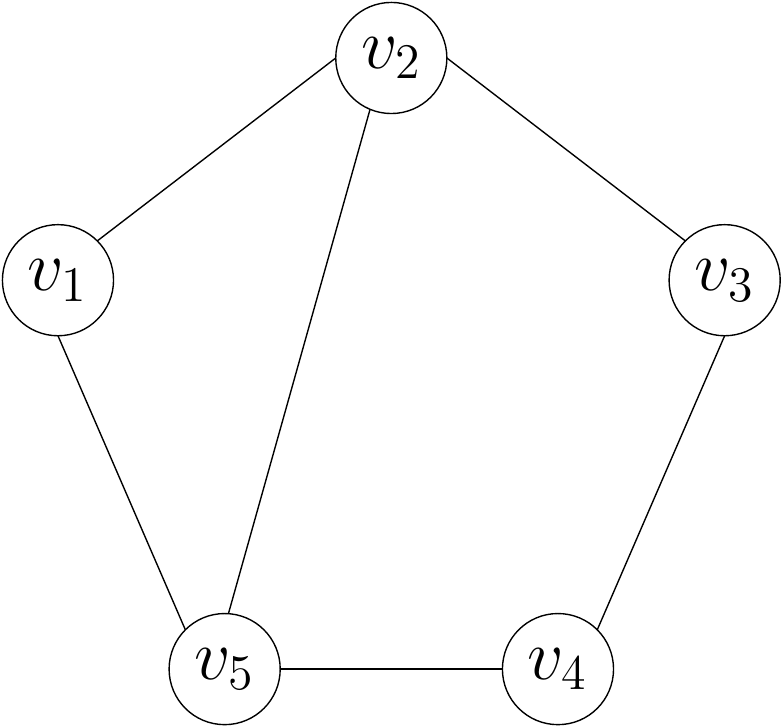}
\caption{\label{fig:ex} An example graph with a vertex ordering.}
\end{fig}

Note that \MakeIndSetZDD does not actually maintain the vertices that are used in a current
independent set $R$ during the ZDD construction.  It is sufficient to maintain a list of vertices
that are left uncovered by {\em some} independent set, since many independent sets may yield the
same set of uncovered vertices.  

\begin{fig}[p]
\caption{\label{fig:steps} A visualization of the steps taken by \texttt{MakeIndSetZDD} to build
$\ZF$ for the example graph given in Figure~\ref{fig:ex}.  Grey nodes and edges have been visited by
a recursive call, but are not yet stored in the ZDD.  Black nodes and edges have been stored in the
ZDD's lookup table.  Bold elements have been inserted in the most recent step of the algorithm;
nodes are labeled in order of insertion.  The set listed to the right of each node is the set $U$
for that recursive call; the notation $[n]$ denotes the set $\{1,2,\hdots,n\}$.}
\begin{subfigure}[t]{0.31\textwidth}
	\centering
    \includegraphics[scale=0.5]{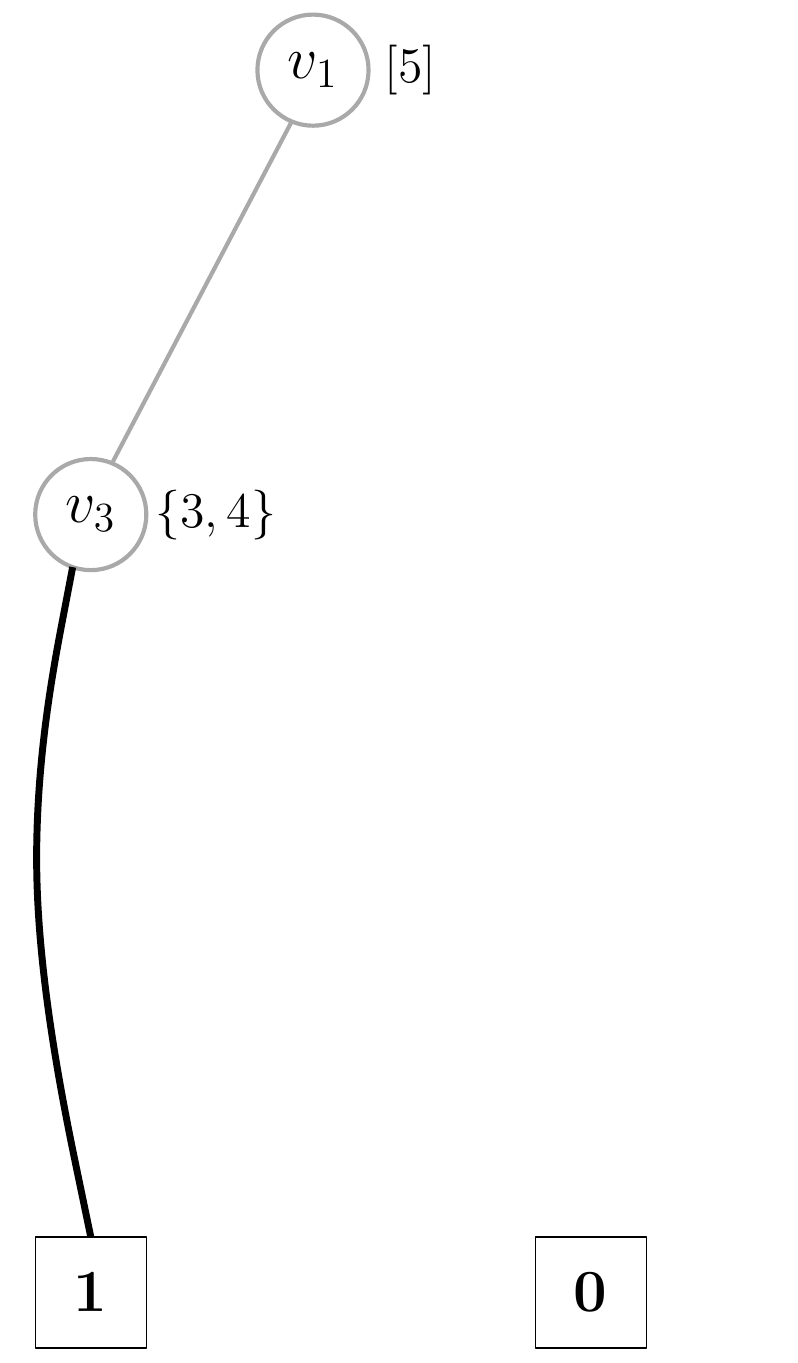}
	\caption{Two recursive calls are made; using vertices $v_1$ and $v_3$ leaves $U$ empty, so
		$\true$ is returned.  No nodes have been inserted into the ZDD at this point.}
	\label{fig:mkzdd1}
\end{subfigure}
\quad
\begin{subfigure}[t]{0.31\textwidth}
	\centering
    \includegraphics[scale=0.5]{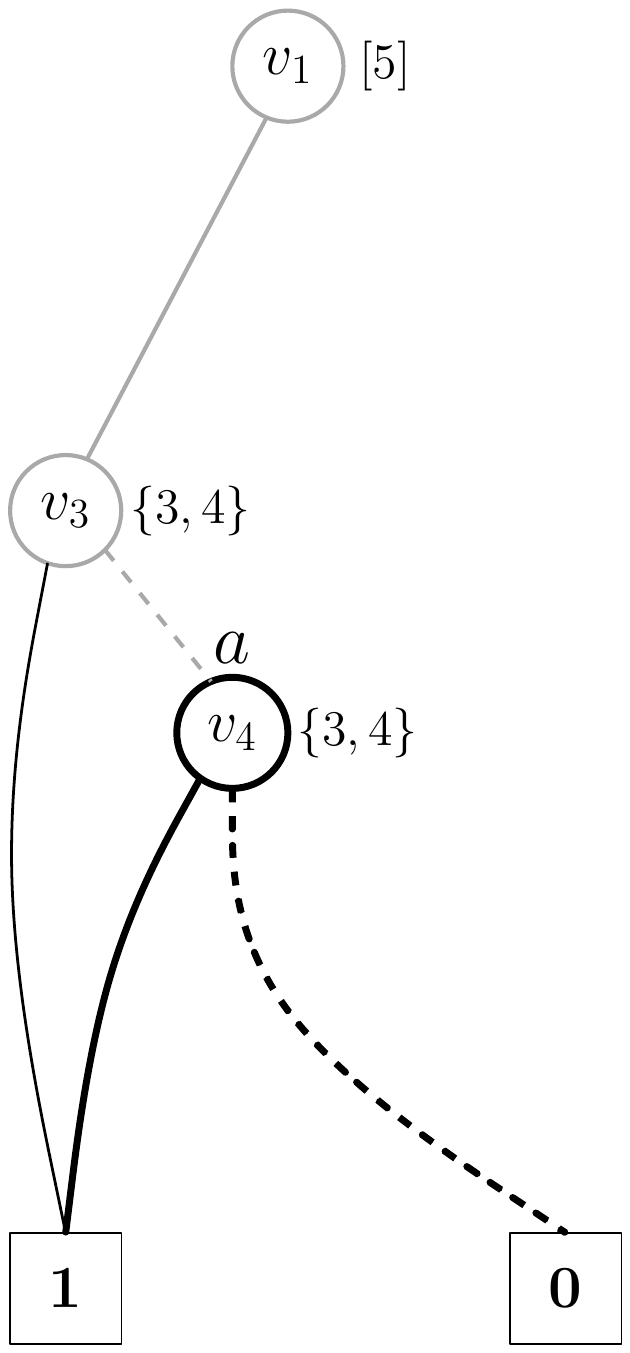}
	\caption{If $v_1$ is used and $v_3$ is not used, $v_4$ must be used.  Node $a$ is the first node
		inserted in the ZDD.}
	\label{fig:mkzdd2}
\end{subfigure}
\quad
\begin{subfigure}[t]{0.31\textwidth}
	\centering
    \includegraphics[scale=0.5]{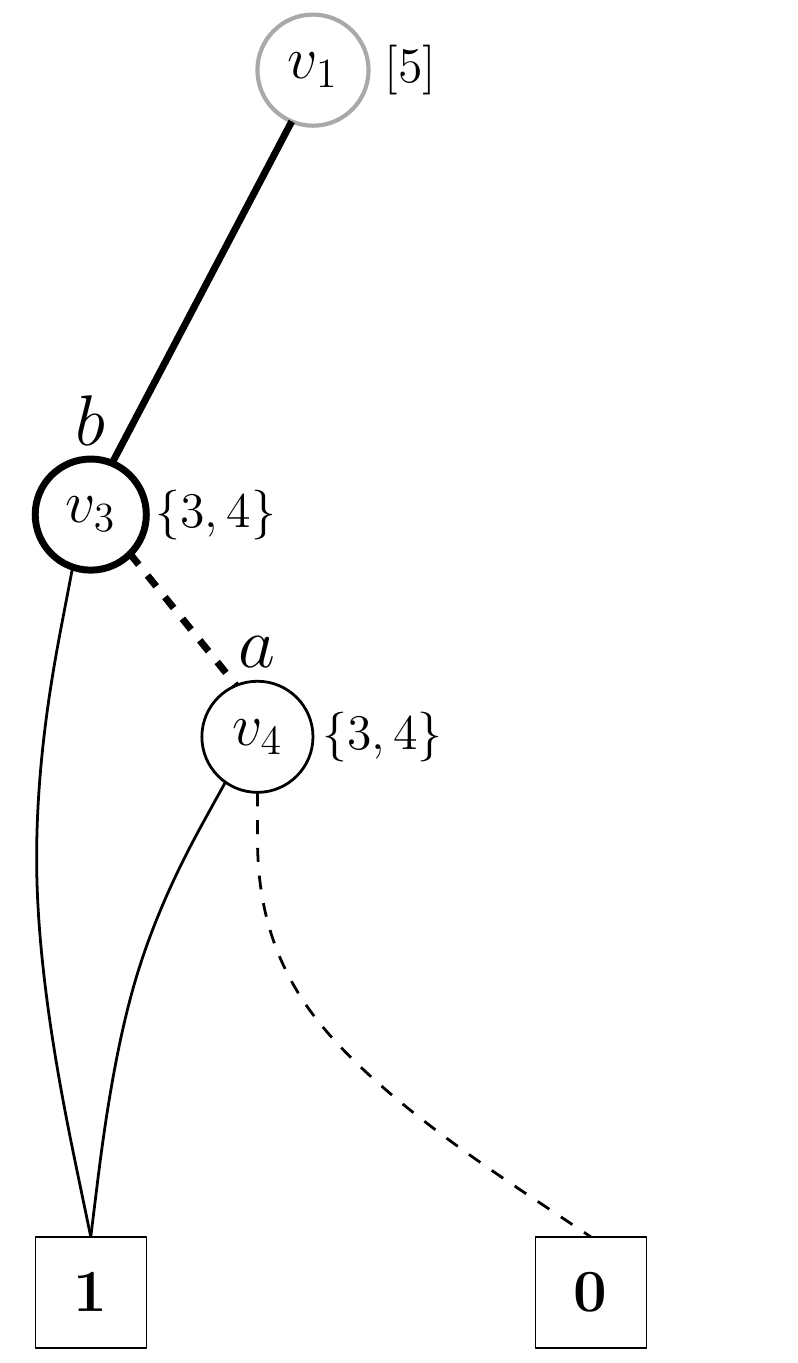}
	\caption{Both children of $b$, the high branch of the root, have been computed, so $b$
		is inserted into the ZDD.}
	\label{fig:mkzdd3}
\end{subfigure}
\\
\begin{subfigure}[t]{0.31\textwidth}
	\centering
    \includegraphics[scale=0.5]{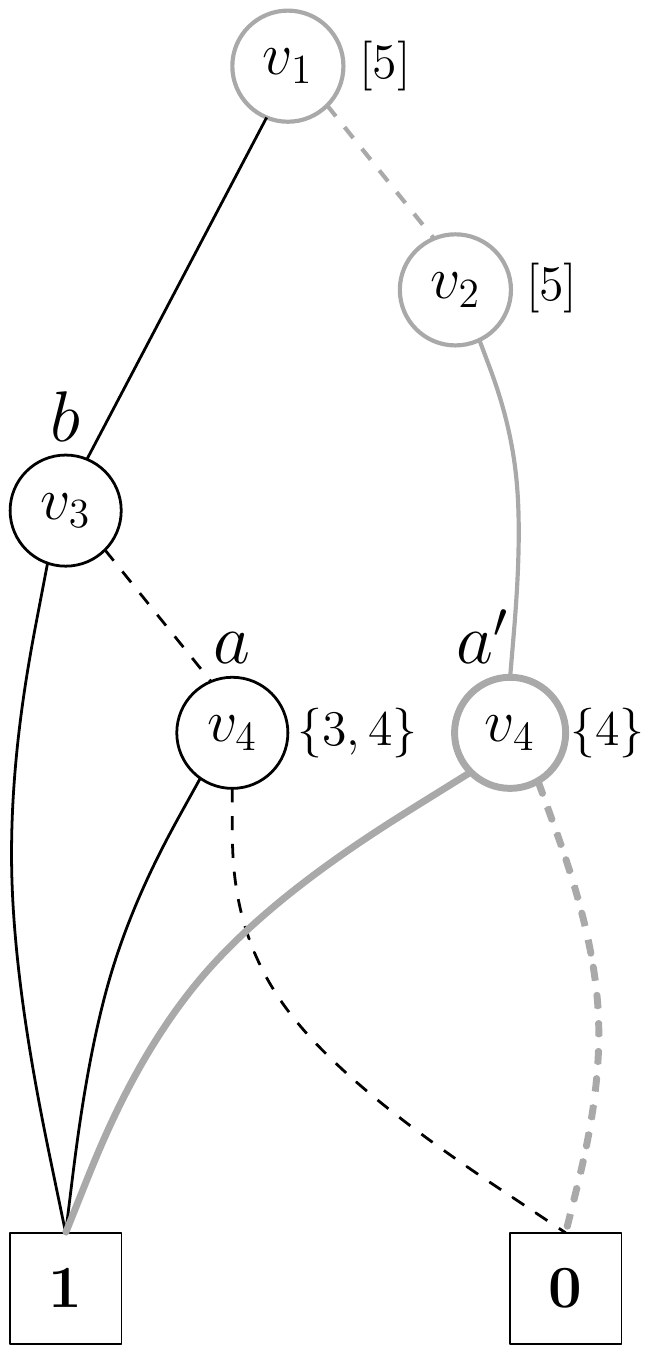}
	\caption{If $v_1$ is not used in an independent set, and $v_2$ is, $v_4$ must also be used to
		ensure maximality.  Node $a'$ is computed as the high child, but is not inserted because it
		is a duplicate of node $a$.} 
	\label{fig:mkzdd4}
\end{subfigure}
\quad
\begin{subfigure}[t]{0.31\textwidth}
	\centering
    \includegraphics[scale=0.5]{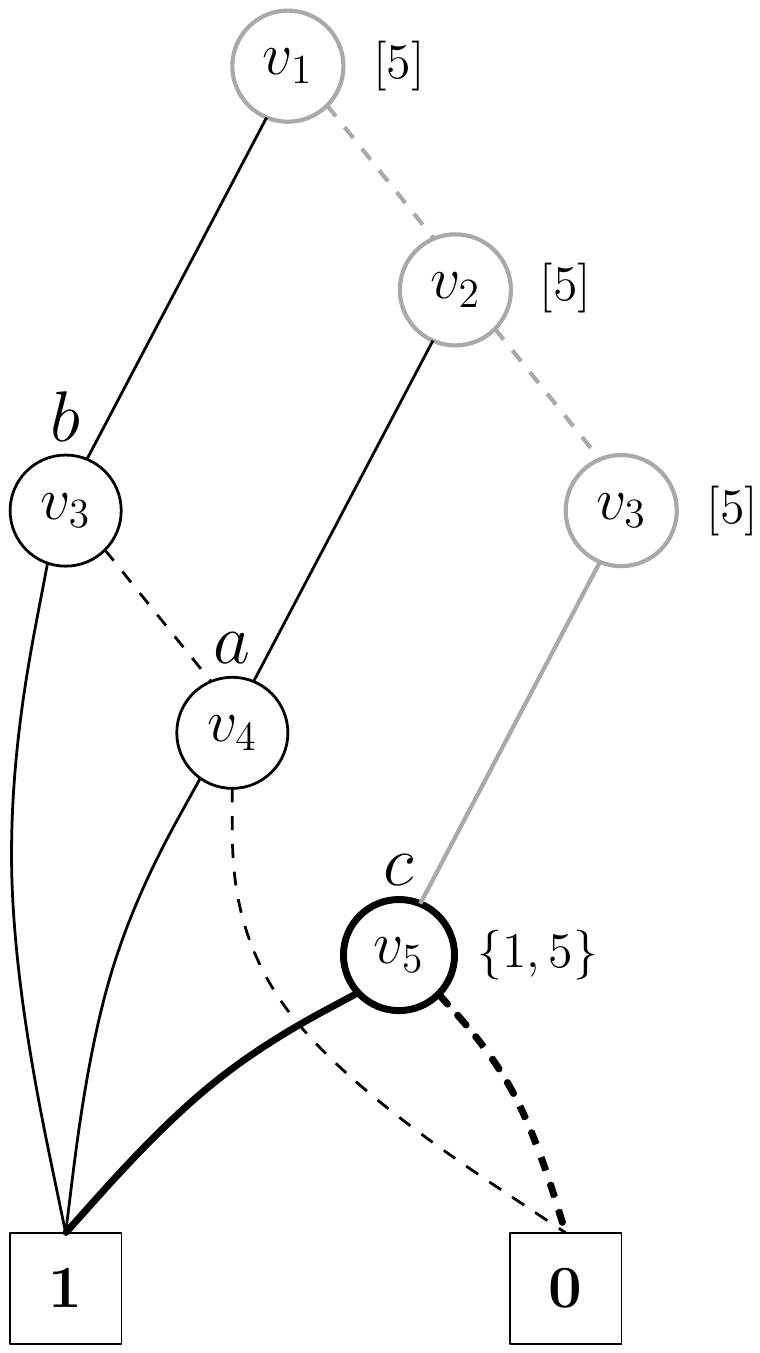}
	\caption{If $v_3$ is the first vertex used in an independent set, $v_5$ must also be used to
		ensure maximality.}
	\label{fig:mkzdd5}
\end{subfigure}
\quad
\begin{subfigure}[t]{0.31\textwidth}
	\centering
    \includegraphics[scale=0.5]{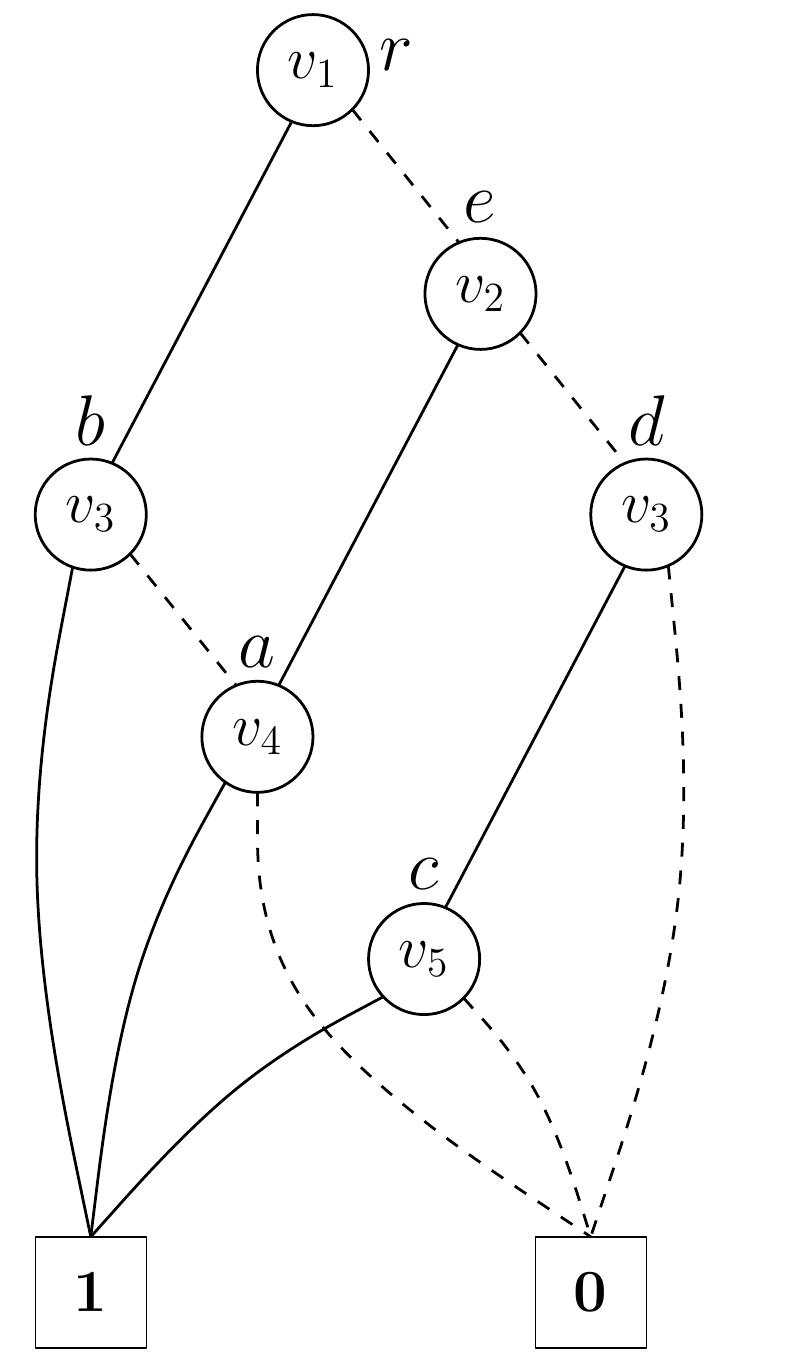}
	\caption{Some vertex in $\{v_1,v_2,v_3\}$ must be used in any maximal independent set of $G$, so
		$\lo{d} = \false$.  All branches are now complete and the algorithm terminates.}
	\label{fig:mkzdd6}
\end{subfigure}
\end{fig}

A number of different ordering heuristics can be applied to the vertex set of $G$ to derive ZDDs of
varying size.  The rule that was empirically found to produce the smallest ZDDs is the maximal path
decomposition rule, which computes a set of paths $P_1,P_2,\hdots,P_q$ such that $P_i$ is maximal in
$G[V - \bigcup_{j=1}^{i-1} P_j]$.  The vertices are then ordered as 
\[v_1^1,v_2^1,\hdots,v_{l_1}^1,v_1^2,v_2^2,\hdots,v_{l_2}^2,v_1^q,v_2^q,\hdots,v_{l_q}^q,\]
where $v_i^j$ is the $\ith$ vertex along the path $P_j$, and $l_j$ is the length of path $P_j$.
\citet{Morrison14zdd} show that the number of nodes associated with the $\kth$ vertex in this
ordering is bounded by the $\kth$ Fibonacci number $F_k$.

\section{Cyclic Best-First Search}\label{sec:cbfs}

As described in Section~\ref{sec:intro}, when using standard integer branching in a branch-and-price
setting, the structure of the search tree can become extremely unbalanced.  In particular, long
chains of assignments that make no progress towards a solution exist, which (if explored) can
dramatically increase the search time.  Moreover, in many cases these long chains appear more
promising than shorter chains which progress towards a solution.  For instance, in a problem with
many covering constraints of the form $\sum_i x_i \geq b$, setting a variable $x_i = 0$ generally
does not change the lower bound much, nor does it restrict the solution considerably since many
other unfixed variables could also satisfy the constraint.

Therefore, standard search strategies such as depth-first search (DFS) or best-first search (BFS) do
not perform particularly well in this setting.  If DFS gets unlucky, it can start exploring some
long chain early in the process which does not improve the incumbent solution but requires a large
amount of search time.  On the other hand, a strategy like BFS which relies on the lower bound to
perform node selection will also perform poorly, since the lower bounds along the long branches will
often be smaller than lower bounds in other parts of the tree.

Historically, the iterative deepening depth-first search (IDFS) strategy \citep{Korf85} has been
used in such settings to prevent the exploration of long chains that do not make progress towards
better incumbents.  However, this paper proposes the use of the cyclic best-first search (CBFS)
strategy as an alternative search strategy that enables the use of lower bound information during
subproblem selection.  The CBFS strategy, originally proposed by \citet{Kao09} (and called
distributed best-first search in their paper), has since been used successfully in a number of
additional settings including two different scheduling problems \citep{Morrison14salb,Sewell12sdst}.

This search strategy can be thought of as a hybrid algorithm between DFS and BFS; the algorithm uses
a measure-of-best function $\mu$ to select the next subproblem to explore (as in BFS), but
repeatedly cycles through a set of labeled \defn{contours} (i.e., a collection of subproblems),
selecting one subproblem from each contour to explore before advancing to the next contour.  The
cyclic behavior can be thought of as a variant of backtracking in DFS.  The contour labels are
simply taken from $\No$ as a way to order the set of contours.  For example, the levels of the search
tree provide a natural contour definition, where subproblems are grouped by their distance from the
root of the tree (see Figure~\ref{fig:contour1}).  

\begin{alg}[b]
\caption{CBFS \label{alg:cbfs}}
Let $i$ be the label of the contour containing the last explored subproblem\;
\If{$\exists C_j \neq \emptyset$ with $j > i$}{Let $j$ be the first index larger than $i$ of a
	non-empty contour\;}
\Else{Let $j$ be the first index in $\{0,1,\hdots,i\}$ of a non-empty contour\;}
return $s \in \argmin_{s' \in C_j} \mu(s')$\;
\end{alg}

Let $C_i$ denote the contour with label $i$; Once a subproblem has been explored from $C_i$, CBFS
chooses the next subproblem for exploration from $C_{i + p}$, where $p = \min \{p' \in \Z^+ \st C_{i
+ p'} \neq \emptyset\}$, and index addition is done modulo $K$ (the largest contour label currently
in use).  The subproblem chosen from this contour is one that minimizes the measure-of-best function
$\mu$.  In contrast, BFS always chooses the subproblem with the best (global) value of $\mu$ to
explore.  Pseudocode for the CBFS strategy is given in Algorithm~\ref{alg:cbfs}; this code is called
in Line~\ref{bpzdd:select} of Algorithm~\ref{alg:bpzdd} to select a new subproblem for exploration.

Previous implementations of CBFS have only used the depth-based contour definition; however, this
contour definition does not produce better performance than DFS or BFS, for the same reason that BFS
performs poorly.  In fact, CBFS can produce worse performance than DFS in some instances of the
graph coloring problem, for instance if DFS gets lucky and finds a good incumbent early.  The key
insight provided here is that other, more complicated, contour definitions are possible which may
dramatically improve the search process for branch-and-price algorithms for graph coloring.

In particular, consider the following contour definition (called the \defn{positive assignment}
definition), which assigns a subproblem to contour $\ell$ if and only if there have been $\ell$
branching decisions made of the form $x_i = 1$ (called a positive assignment).  Using this contour
definition significantly restructures the order in which subproblems are selected for exploration
(see Figure~\ref{fig:contour2}).  In particular, using this definition prohibits the immediate
exploration of a child subproblem which assigns $x_i = 0$, even if the lower bound at this child is
better than the lower bound at the $x_i = 1$ child.  In this way, the search is weighted towards
exploration of subproblems that make positive assignments, which serves to counterbalance the
effects of a lopsided search tree.  In essence, the positive assignments can be thought of as
discrepancies in limited discrepancy search \citep{Harvey95,Korf96}: most of the maximal independent
sets in $G$ will not be used, but a few, the discrepancies, will be.

\begin{fig}[t]
\caption{\label{fig:contours} Subproblems in a search tree with two different contour functions.
Dashed lines indicate an assignment of $0$ to the variable at that subproblem.  The number in the
center of each node is the label of the contour that subproblem is assigned to.}
\begin{subfigure}[t]{0.48\textwidth}
	\includegraphics[scale=0.4]{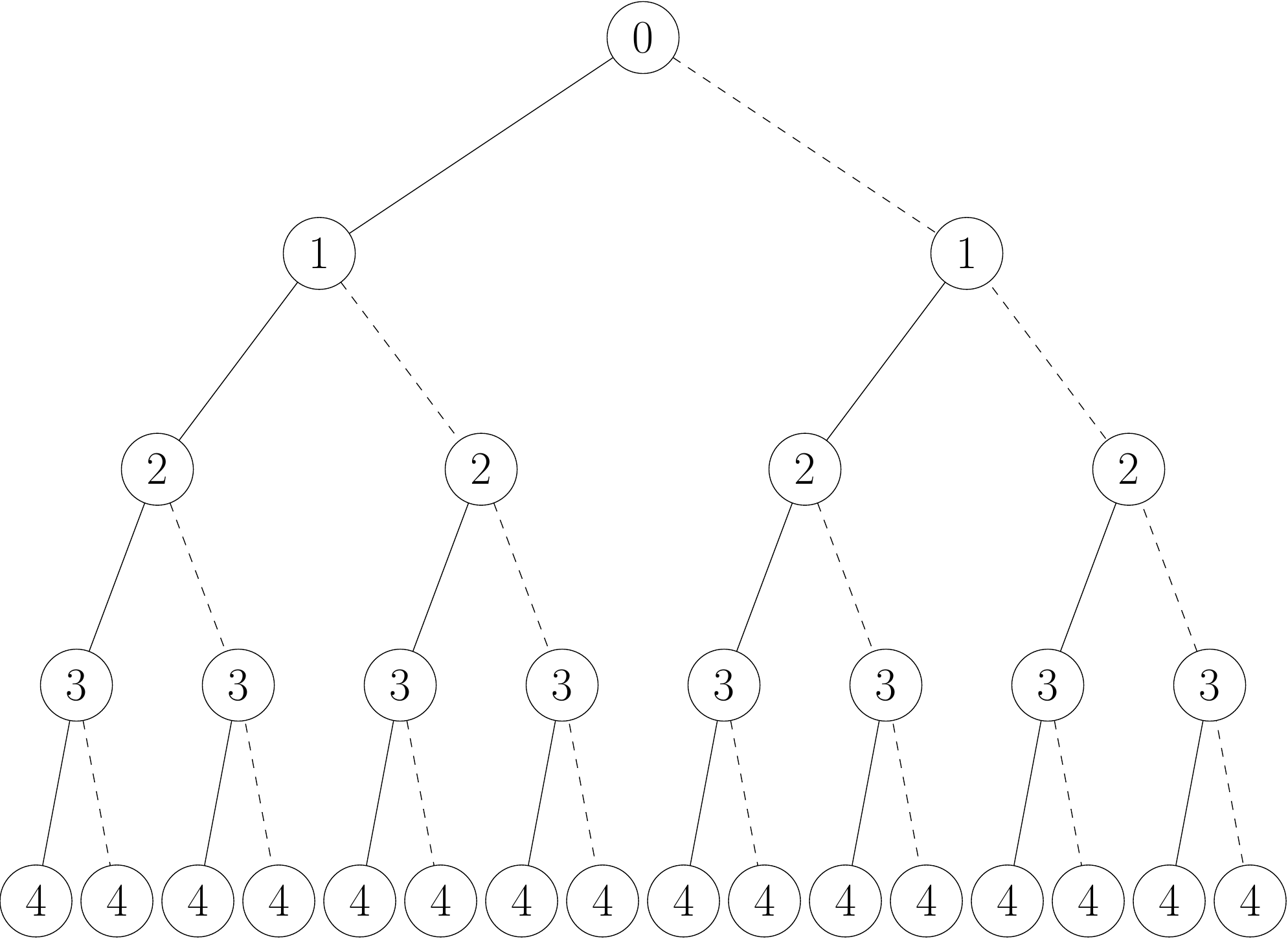}
	\caption{The depth contour definition}
	\label{fig:contour1}
\end{subfigure}
\quad
\begin{subfigure}[t]{0.48\textwidth}
	\includegraphics[scale=0.4]{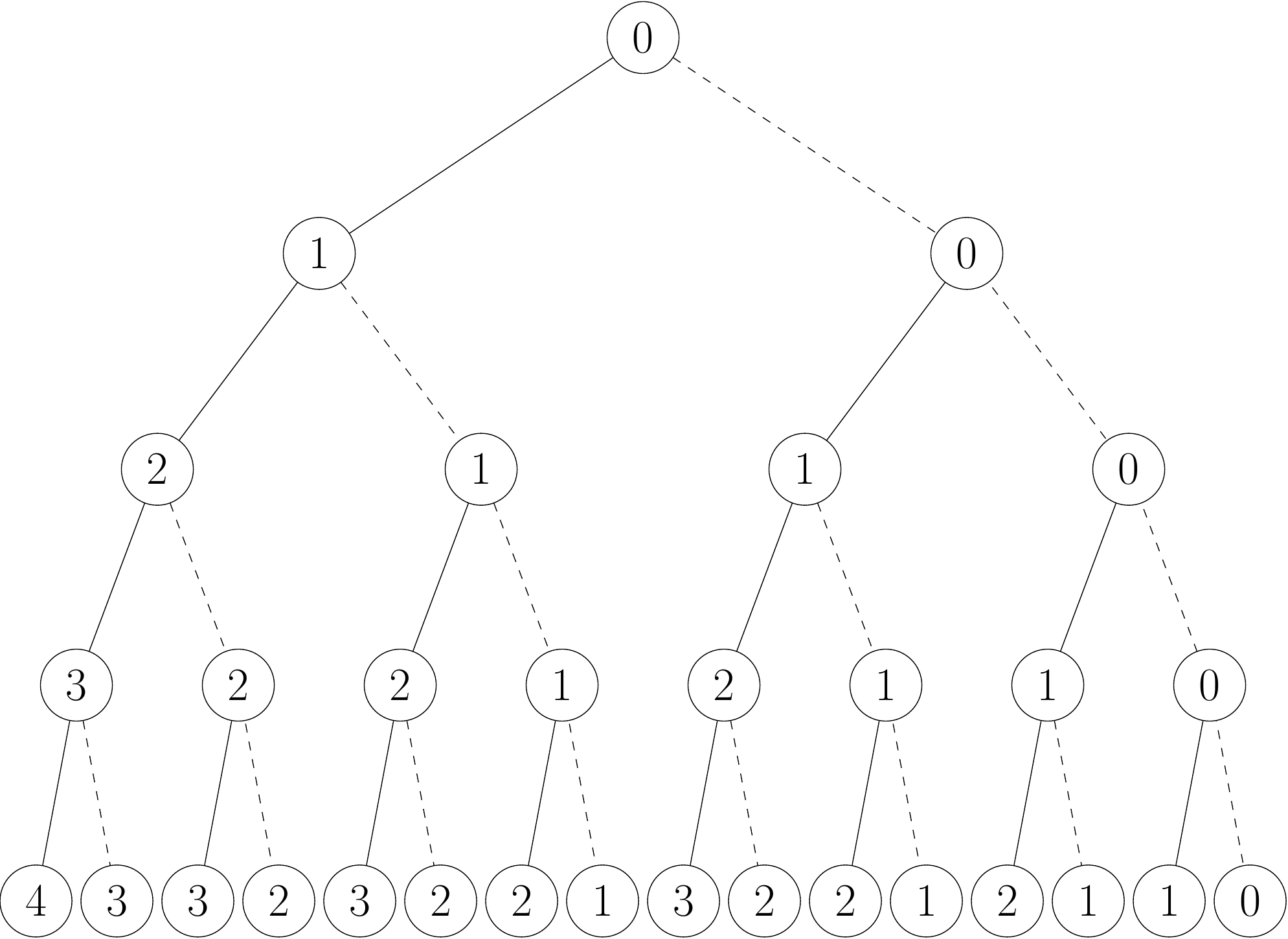}
	\caption{The positive assignment contour definition}
	\label{fig:contour2}
\end{subfigure}
\end{fig}

\section{Computational Results}\label{sec:compres}

A branch-and-price algorithm for the graph coloring problem was implemented using a ZDD to solve the
pricing problem, together with the CBFS strategy for subproblem exploration, and computational
experiments were run on a subset of the instances from the DIMACS graph coloring challenge
\citep{Johnson96, Trick05}.  This section describes some implementation details for this program,
called \BPZ, and discusses the results of these experiments and a comparison to the best algorithms
in the literature.

\subsection{Initialization and Preprocessing}\label{sec:init}

To reduce the size of problem instances, \BPZ uses a standard preprocessing technique: a search is
done to find a large clique $C$ in the graph, and any vertex $v \in V$ with degree less than
$|C|$ is removed.  Since a valid coloring for $G$ must use at least $|C|$ colors, at least one color
exists in any proper coloring that is not assigned to any neighbor of $v$; thus, any proper coloring
of $G - v$ can be extended to $G$ without increasing the number of colors used \citep{Mendez06}.
A branch-and-bound search is employed in a heuristic manner to find an initial large clique.  The
clique $C$ can also be used to prove optimality -- if a proper coloring of $G$ is found that uses
exactly $|C|$ colors, this coloring must be optimal.

To initialize \BPZ, a starting pool of independent sets needs to be generated.  A modified version
of the initialization procedure described in \citet{Malaguti08} is used for this purpose.  Their
algorithm employs a 2-phase approach to find good initial solutions.  In the first phase, a genetic
algorithm combined with a local search rule searches for valid $k$-colorings of the graph for some
input parameter $k$.  If a valid $k$-coloring is found, then the procedure is iteratively called
with successively smaller values of $k$ until a user-specified time limit is reached.  The second
phase takes the best solution found in phase 1 and applies a covering heuristic to improve the
solution further.  \BPZ uses a similar procedure to generate its initial pool of independent sets
for the RMP, which only runs the first phase of the algorithm described by \citet{Malaguti08}.

Any column generated by the initialization routine can be added to the initial pool $\mathcal{S}'$
for the RMP.  However, the initialization procedure often generates a large number of sets; thus, it
is necessary to reduce the size of the initial pool.  To this end, the RMP is solved once with only
the sets used by the best available coloring to get initial dual prices.  Only the generated sets
with a price above $0.8$ are included in $\mathcal{S}'$.  This rule includes all sets with negative
or close-to-negative reduced cost in $\mathcal{S}'$, since these sets are more likely to improve
upon the LP solution to the RMP in early stages of the search.

\subsection{Results from the DIMACS Database}\label{sec:dimacs}

\BPZ was implemented in C++ and used CPLEX 12.5 with default settings to solve the RMP; all
computational experiments described in this section were performed on a desktop machine with an
Intel Core i7-930 2.8GHz quad-core processor and 12 GB of available memory.  The branch-and-price
algorithm utilized only a single processor core; however, CPLEX operates in parallel by default.
All times reported here are aggregated over all cores.  For the sake of comparison with the results
obtained with the MMT algorithm, the {\em dfmax} benchmark program was run on the r500.5 instance
provided by \citet{Trick05}.  The computer used for these experiments took 6.60s CPU time to solve
this benchmark instance.

\looseness=-1
Comparisons were made against four different branch-and-price algorithms available in the
literature.  First, \citet{Malaguti11} give an exact algorithm for the graph coloring problem that
uses an improved initialization heuristic, together with extensive computational results.  The
compared results were obtained using standard $0-1$ branching instead of edge branching.  Secondly,
\citet{Gualandi12} describe a branch-and-price solver for graph coloring that uses constraint
programming techniques to solve the pricing problem; their implementation uses the edge branching
rule.  \citet{Morrison14wide} provide extensive computational results using a wide branching
technique, which modifies the standard $0-1$ branching rule to allow multiple children to be
generated from a subproblem in the search tree.  Finally, \citet{Held12} provide a method for
computing a numerically safe lower bound for graph coloring, which they embed inside a
branch-and-price solver.  Using this algorithm, they are able to prove new lower bounds for a number
of unsolved instances.  

\begin{tab}
\begin{tabular}{r|p{11cm}}
Instance & Name of the tested instance\\
$n$ & Number of vertices in the instance\\
$m$ & Number of edges in the instance\\
$\chi$ & Chromatic number of the instance, if known\\
$\LB$ & Lower bound found by \BPZ\\
$\UB$ & Best solution found by \BPZ (if $t < 10hrs$, this value is optimal)\\
$t_Z$ & Time needed to construct the maximal independent set ZDD\\
$t$ & Time spent in the branch-and-price phase of the algorithm\\
$t_{\textrm{MMT-init}}$ & Total initialization time used by \citet{Malaguti11}\\
$t_{\textrm{MMT}}$ & Adjusted time to verify optimality by \citet{Malaguti11}\\
$t_{\textrm{wide}}$ & Adjusted time to verify optimality by the wide branching solver of
	\citet{Morrison14wide}\\ 
$t_{\textrm{CP-BnP}}$ & Adjusted time to verify optimality by the branch-and-price solver of
	\citet{Gualandi12}\\ 
exp & Number of nodes explored in the search tree\\
id & Number of nodes identified in the search tree\\
$Z_i$ & Initial size of the ZDD\\
$Z_f$ & Final size of the ZDD\\
\% change & Percent change in size of the ZDD over the course of the algorithm\\
$|\textrm{col}|$ & Number of columns generated over the course of the algorithm\\
$t_{\textrm{price}}$ & Time spent solving all pricing problems over the course of the
	algorithm\\
\end{tabular}
\caption{\label{tab:cols} Notation used for computational results data (Tables~\ref{tab:zdd}
	and~\ref{tab:zdd2}).}
\end{tab}

\looseness=-1
Experiments were run on 40 instances from the DIMACS instance database \citep{Trick05}.  Experiments
were not run on easy instances (those for which the lower bound at the root is sufficient to prove
optimality), since these instances do not demonstrate the effectiveness of the ZDD data structure
for solving the pricing problem in the presence of branching constraints.  The remaining instances
were chosen to span a range of difficulty, including ones that are easily solved to optimality by
all algorithms in the literature, and others for which no algorithm has yet been able to verify
optimality.  In addition, experiments were run on 7 additional instances taken from
\citet{Gualandi12}.

\input{data}

A time limit of 10 hours was imposed for all experiments, and the ZDD size was limited to
$100\,000\,000$ nodes.  The initialization procedure from Section~\ref{sec:init} was run for 100
seconds for each instance to generate an initial pool; this did not contribute to the 10-hour time
limit.  Of the 40 instances tested, most were extremely difficult, and could not be solved by any
algorithm within the 10-hour time limit.  Data for all 47 instances are shown in
Table~\ref{tab:zdd}, and notation used in this table is given in Table~\ref{tab:cols}.  To provide
the most fair comparison between different computational platforms, all running times are scaled
according to the benchmark value of the {\em dfmax} utility reported.  In most cases, \citet{Held12}
is concerned with computing lower bounds instead of total solution times, so data from this paper
are omitted from Table~\ref{tab:zdd} (though they are still discussed in the sequel).

Computational comparisons across differing models, source code, and platforms is notoriously
difficult, and these data are no exception.  The MMT algorithm uses an initialization procedure with
a variable running time, depending on the difficulty of the problem.  However, \BPZ is given a flat
100 seconds of initialization, plus the length of time required to build the ZDD; moreover, the
initialization procedure used by \BPZ is strictly weaker than that of the MMT algorithm, since it
only uses the first phase of a two-phase procedure.  This choice was made to highlight the
performance benefits of using ZDDs; in principle, if \BPZ had access to the full initialization
procedure used by the MMT algorithm, its results would be even better.  Therefore, to produce the
fairest comparison between the algorithms, the compared running times are solely the time spent on
the branch-and-price algorithm, not including the initialization time.

\BPZ was able to find and verify optimality for 15 of the 47 instances tested.  In four cases, \BPZ
is able to find and verify optimality at least an order of magnitude faster than any competing
algorithm.  Additionally, \BPZ is able to verify optimality for three new instances
(\texttt{1-FullIns\usc4}, \texttt{r1000.5}, and \texttt{flat\usc300\usc0}) that have not been solved
previously by branch-and-price algorithms in the literature (however, in the case of
\texttt{r1000.5}, the ZDD construction took longer than 10 hours).  One other instance,
\texttt{DSJC250.9}, has only been solved by the branch-and-price solver of \citet{Held12}; their
algorithm found a solution in 8685 (adjusted) CPU seconds.

It was observed that modifying the initial pool size can dramatically improve the running time of
\BPZ; for example, running the initialization procedure for 6100 seconds (the default initialization
time limit in \citet{Malaguti11} and \citet{Morrison14wide}) allows \BPZ to solve \texttt{DSJC125.5}
in 31 seconds.  Similarly, running the initialization procedure for only 3 seconds allows \BPZ to
solve \texttt{queen9\usc9} in 2.3 seconds.  This is explained by noting that a large initial pool
can slow down the LP solver for the RMP.

Finally, there are five instances which were solved substantially faster by the MMT graph coloring
solver than by \BPZ; however, four of these instances were solved at the root node by the MMT
solver due to a better initialization procedure, and so do not provide a useful comparison against
\BPZ.  This leaves only one instance (\verb+queen11_11+) for which some other algorithm
substantially outperforms \BPZ; for this instance, the lower bound is equal to the optimal objective
value, which means the search can be terminated as soon as an optimal solution is found.  

Data were collected regarding the average length of time needed to solve the pricing problem for
each instance, as well as the growth in size of the ZDD over the course of the algorithm.  The
average growth in size of a ZDD for any problem was 14\%, with a standard deviation of 27\%.  In
one case, the size of the ZDD nearly doubled, at 94\% growth; however, even in this case, the length
of time needed to solve the pricing problem was not impacted substantially.  In most cases when the
ZDD could be fully constructed, the length of time needed to solve one iteration of the pricing
problem was under a second.  

\input{data2}

Details regarding these data are shown in Table~\ref{tab:zdd2}, with column headings again given in
Table~\ref{tab:cols}.  Here, note that the number of nodes explored in the tree (column 4) is the
total number of nodes at which children were generated.  The number of nodes identified in the tree
(column 5) includes those nodes for which the RMP was solved, but could be pruned away before
generating children.  Finally, the last column in this table shows the total CPU time taken to solve
every pricing problem for each instance.  As the bulk of the work to solve an instance is in solving
the RMP and solving the pricing problem, the time needed to solve the RMP for a particular instance
is approximately $t$ (column 8, Table~\ref{tab:zdd}) minus $t_{\textrm{price}}$.

In order to assess the efficacy of the search strategy, additional computational experiments were
run against the 12 instances from the DIMACS database that \BPZ was able to solve within the time
limit.  \BPZ was modified for these tests to use DFS instead of CBFS, and the results are presented
in Table~\ref{tab:dfs}.  

\begin{table}[h]
\begin{center}
\begin{tabular}{l|rr}
	Instance & $t_{DFS}$ & $t_{CBFS}$ \\\hline
	\texttt{DSJC125.5} & 7689.02 & \cellcolor[gray]{0.8} 284.92 \\
	\texttt{DSJC125.9} & \cellcolor[gray]{0.8} 0.17 & 0.21\\
	\texttt{DSJC250.9} & 1587.94 & \cellcolor[gray]{0.8} 649.90\\
	\texttt{DSJR500.1c} & \cellcolor[gray]{0.8} 0.07 & 0.10\\
	\texttt{DSJR500.5} & 103.70 & \cellcolor[gray]{0.8} 102.75\\
	\texttt{queen9\usc{}9} & 11.31 & \cellcolor[gray]{0.8} 9.33\\
	\texttt{queen10\usc{}10} & \cellcolor[gray]{0.8} 134.07 & 140.76\\
	\texttt{queen11\usc{}11} & 32057.94 & \cellcolor[gray]{0.8} 14354.16 \\
	\texttt{myciel4} & \cellcolor[gray]{0.8} 0.08 & 0.09\\
	\texttt{myciel5} & \cellcolor[gray]{0.8} 244.37 & 392.21\\
	\texttt{1-FullIns\usc{}4} & \cellcolor[gray]{0.8} 59.04 & 122.58\\ 
\end{tabular}
\caption{\label{tab:dfs} A comparison of the running times of \texttt{B\&P+ZDD} using DFS and CBFS.}
\end{center}
\end{table}

In these experiments, there are 3 instances in which CBFS (with the positive assignment rule)
significantly outperforms DFS, and 2 instances in which DFS outperforms CBFS.  For the remaining 7
instances, both algorithms perform within a few seconds of each other.  For the three instances
where CBFS outperformed DFS, the improvement was an order of magnitude in one case, and over twice
as fast in the other two cases.  Therefore, it was concluded that in general, CBFS with the positive
assignment rule is a better choice of branching strategy for the graph coloring problem than DFS.

\section{Conclusions and Future Work}\label{sec:end}

This paper presents a framework for using standard integer branching in conjunction with
branch-and-price algorithms; this framework solves the pricing problem using a zero-suppressed
binary decision diagram that is constructed during a preprocessing phase.  When new columns are
generated, they are restricted from generation by the ZDD a second time; this allows the constrained
pricing problem to be solved exactly at every iteration of the algorithm.  Using this technique
combined with a new contour definition for the cyclic best-first search strategy to counterbalance
the resulting lopsided search tree, the standard integer branching scheme can be used in conjunction
with a branch-and-price algorithm, which yields a much faster and more direct solution method in
many cases.  Computational results were presented showing that a branch-and-price algorithm
implementation for the graph coloring problem in some cases outperforms other branch-and-price 
graph coloring solvers in the literature.  In several cases, this performance is an improvement of
an order of magnitude or more, though an exact comparison is difficult due to differences in
initialization.

A number of future research directions exist for this method; firstly, this paper proposes a ZDD
algorithm for the graph coloring problem.  However, there is nothing specific to graph coloring in
Sections~\ref{sec:defn} and \ref{sec:restrict}; in fact, the techniques described here are general,
and could be applied to other problems.  Some preliminary results using ZDDs with the generalized
assignment problem are described in \citet{Morrison14thesis}, but more work must be done to show
their effectiveness on other types of problems.  Moreover, ZDDs can also be used even if the
branch-and-price solver does not require the solution of the constrained pricing problem (for
instance, the robust branch-and-price-and-cut algorithm of \citealp{deAragao03}).  In these
settings, the ZDD does not need to have restrictions imposed via \RestrictSet when a new variable is
generated; however, they may still provide benefits, since the ZDD is able to produce a variable of
most negative reduced cost at every iteration of column generation.  Thus, additional research can
be done to study how ZDDs interact with other established branch-and-price methods.

Secondly, research can be performed to determine the best way to use ZDDs when the entire data
structure will not fit in memory; one proposed idea uses approximate ZDDs (described in
\citealp{Bergman12b}) to solve the pricing problem in these settings.  An approximate ZDD is a
width-constrained ZDD that does not eliminate any valid solutions to the pricing problem, but may
accept some inputs that are not valid solutions to the pricing problem.  In this setting, a
post-generation check can be performed to see if the ZDD produced a set that is a valid solution to
the pricing problem; if not, the \RestrictSet routine can be called with the erroneous solution to
prevent it from being generated again.  If the approximate ZDD can be constructed in an appropriate
fashion, it is hypothesized that invalid solutions will be generated relatively infrequently, and
the algorithm will not suffer much loss of efficiency.

A final important question addresses the addition and removal of restrictions from the pricing
problem ZDD.  Many standard branch-and-price algorithms will generate multiple columns in between
each intermediate solution of the RMP.  It has been observed that this can improve algorithm
performance; thus, an interesting research direction may be to modify the ZDD algorithm to generate
multiple columns in a single pass through the data structure.  Moreover, in many cases the variable
pool for branch-and-price algorithms may grow quite large over the course of the algorithm.  Since
the size of this pool directly impacts the solution time for the RMP, and most of the elements of
this pool are never used in any optimal solutions, most branch-and-price solvers will prune the pool
by deleting variables with very large positive cost.  In such a setting, it is necessary to modify
$\ZF$ to allow removed variables to be generated again; these variables can be stored in an
auxiliary pool to be scanned before the ZDD is queried.  Thus, future research should analyze the
effects of such a modification.

\section*{Acknowledgments}

The computational results reported were obtained at the Simulation and Optimization Laboratory at
the University of Illinois, Urbana-Champaign. This research has been supported in part by the Air
Force Office of Scientific Research (FA9550-10-1-0387, FA9550-15-1-0100), the National Defense
Science \& Engineering Graduate Fellowship (NDSEG) Program, and the National Science Foundation
Graduate Research Fellowship Program under Grant Number DGE-1144245.  Any opinion, findings, and
conclusions or recommendations expressed in this material are those of the author(s) and do not
necessarily reflect the views of the United States Government. The authors would also like to thank
Jason Sauppe for conversations and help regarding the branch-and-price implementation, as well as
Enrico Malaguti for the source code used in our initialization procedure.  Finally, the authors
thank the two anonymous referees and the associate editor for comments which resulted in a
much-improved draft of this paper.

\input{colorzdd.bbl}
\end{document}

%% file: data.tex
\begin{table}
\scriptsize
\begin{tabular}{lrrr|rrrrrrrr}
Instance & $n$  &  $m$  &  $\chi$  &  \LB  &  \UB  &  $t_Z$  &  $t$ &$T_{\textrm{MMT-init}}$& $t_{\textrm{MMT}}$ & $t_{\textrm{wide}}$ & $t_{\textrm{CP-BnP}}$\\\hline
\texttt{DSJC125.5}  &125&3891&17&16&17&0.47&284.92&6100&17019.33&\cellcolor[gray]{0.8} 225.21&14372.75\\
\texttt{DSJC125.9}  &125&6961&44&43&44&\cellcolor[gray]{0.8}0&\cellcolor[gray]{0.8}0.21&6100&3674.22&1.02&33.23\\
\texttt{DSJC250.5}  &250&15668&  ?  &26&28&29.55& \textgreater10hrs &6100& \textgreater10hrs & \textgreater10hrs & \textgreater10hrs\\
\texttt{DSJC250.9}  &250&27897&72&71&72&\cellcolor[gray]{0.8}0.04&\cellcolor[gray]{0.8}649.90&6100& \textgreater10hrs & \textgreater10hrs & \textgreater10hrs\\
\texttt{DSJC500.5}  &500&62624&  ?  &43&53&3458.5& \textgreater10hrs &6100& \textgreater10hrs & \textgreater10hrs & -\\
\texttt{DSJC500.9}  &500&112437&  ?  &123&129&0.38& \textgreater10hrs &6100& \textgreater10hrs & \textgreater10hrs & \textgreater10hrs\\
\texttt{DSJC1000.5}  &1000&249826&  ?  &10&108&6608.86& oom &6100& \textgreater10hrs &  & -\\
\texttt{DSJC1000.9}  &1000&449449&  ?  &215&228&5.64& \textgreater10hrs &6100& \textgreater10hrs & \textgreater10hrs & -\\
\texttt{DSJR500.1c}  &500&121275&85&85&85&\cellcolor[gray]{0.8}0.14&\cellcolor[gray]{0.8} 0.10&6100&272.01&1.29&0.60\\
\texttt{DSJR500.5}  &500&58862&122&122&122&689.48& 102.75&6100&322.6&6862.28& \textgreater10hrs\\
\texttt{le450\usc{}25c}  &450&17343&25&25&28&19889.92& oom &6100&\cellcolor[gray]{0.8} init & \textgreater10hrs & -\\
\texttt{le450\usc{}25d}  &450&17425&25&25&28&16262.45& oom &6100&\cellcolor[gray]{0.8} init & \textgreater10hrs & -\\
\texttt{queen9\usc{}9}  &81&1056&10&9&10&\cellcolor[gray]{0.8}0.44&\cellcolor[gray]{0.8} 9.33&3&34.51&20.48&74.00\\
\texttt{queen10\usc{}10}  &100&2940&11&10&11&\cellcolor[gray]{0.8}4.09&\cellcolor[gray]{0.8} 140.76&3&647.65&587.32&26393.20\\
\texttt{queen11\usc{}11}  &121&3960&11&11&11&33.1&14354.16&3&\cellcolor[gray]{0.8} 1759.08&19208.45&21858.50\\
\texttt{queen12\usc{}12}  &144&5192&12&12&13&297.91& \textgreater10hrs &3& \textgreater10hrs & \textgreater10hrs & -\\
\texttt{queen13\usc{}13}  &169&6656&13&13&14&2739.39& \textgreater10hrs &100& \textgreater10hrs & \textgreater10hrs & -\\
\texttt{queen14\usc{}14}  &196&8372&14&14&15&3139.2& \textgreater10hrs &100& \textgreater10hrs & \textgreater10hrs & -\\
\texttt{queen15\usc{}15}  &225&10360&15&15&16&3235.39& \textgreater10hrs &100& \textgreater10hrs & \textgreater10hrs & -\\
\texttt{queen16\usc{}16}  &256&12640&16&16&18&3527.63& \textgreater10hrs &100& \textgreater10hrs & \textgreater10hrs & -\\
\texttt{myciel3}  &11&23&4&3&4&0&0&3&0&0.01& -\\
\texttt{myciel4}  &20&71&5&4&5&\cellcolor[gray]{0.8}0&\cellcolor[gray]{0.8} 0.09&3&111.26&0.47& -\\
\texttt{myciel5}  &47&236&6&4&6&\cellcolor[gray]{0.8}0.01&\cellcolor[gray]{0.8} 392.21&3& -&3207.63& -\\
\texttt{myciel6}  &95&755&7&4&7&0.76& \textgreater10hrs &3& \textgreater10hrs & \textgreater10hrs & -\\
\texttt{myciel7}  &191&2360&8&5&8&2198& \textgreater10hrs &3& \textgreater10hrs & \textgreater10hrs & -\\
\texttt{1-Insertions\usc{}4}  &67&232&5&3&5&0.69& \textgreater10hrs &3& \textgreater10hrs & \textgreater10hrs & -\\
\texttt{1-Insertions\usc{}5}  &202&1227&  ?  &2&6&23708.86&\textgreater10hrs&3& \textgreater10hrs & \textgreater10hrs & -\\
\texttt{2-Insertions\usc{}4}  &149&541&  ?  &2&5& \textgreater10hrs & - &3& \textgreater10hrs & \textgreater10hrs & -\\
\texttt{3-Insertions\usc{}3}  &56&110&4&3&4&1.8& \textgreater10hrs &3& \textgreater10hrs & \textgreater10hrs & -\\
\texttt{3-Insertions\usc{}4}  &281&1046&  ?  &2&5&22684.73& \textgreater10hrs &3& \textgreater10hrs & \textgreater10hrs & -\\
\texttt{4-Insertions\usc{}3}  &79&156&4&3&4&314.85& \textgreater10hrs &3& \textgreater10hrs & \textgreater10hrs & -\\
\texttt{1-FullIns\usc{}4}  &93&593&5&4&5&\cellcolor[gray]{0.8}8.76&\cellcolor[gray]{0.8} 122.58&3& \textgreater10hrs & \textgreater10hrs & -\\
\texttt{1-FullIns\usc{}5}  &282&3247&6&3&6& \textgreater10hrs & - &3& \textgreater10hrs & \textgreater10hrs & -\\
\texttt{2-FullIns\usc{}4}  &212&1621&6&0&0& \textgreater10hrs & - &3& \textgreater10hrs & \textgreater10hrs & -\\
\texttt{2-FullIns\usc{}5}  &852&12201&7&4&7& \textgreater10hrs & - &3& \textgreater10hrs & \textgreater10hrs & -\\
\texttt{3-FullIns\usc{}4}  &405&3524&7&5&7& \textgreater10hrs & - &3& \textgreater10hrs & \textgreater10hrs & -\\
\texttt{4-FullIns\usc{}4} &690&6650&8&0&0& \textgreater10hrs & - &3& \textgreater10hrs & \textgreater10hrs & -\\
\texttt{latin\usc{}square\usc{}10}  &900&307350&  ?  &90&100&36.4& \textgreater10hrs &6100& \textgreater10hrs & \textgreater10hrs & -\\
\texttt{qg.order30} &900&26100&30&0&0& \textgreater10hrs & - &3&\cellcolor[gray]{0.8} 0.19& \textgreater10hrs & -\\
\texttt{wap06}  &947&43571&40&0&0& \textgreater10hrs & - &170&\cellcolor[gray]{0.8} 165.00& \textgreater10hrs & -\\
\texttt{r250.5} &250&14849&65&65&65&14.63&7.15& - & - & - &\cellcolor[gray]{0.8} 6.80\\
\texttt{r1000.1c} &1000&485090& ? &96&98&2.29& \textgreater10hrs & - & - & - & \textgreater10hrs\\
\texttt{r1000.5} &1000&238267&234&234&234&43020.09& 4690.16& - & - & - & \textgreater10hrs\\
\texttt{flat300\usc{}28\usc{}0} &300&21695&28&28&28&\cellcolor[gray]{0.8}144.57&\cellcolor[gray]{0.8} 19883.45& - & - & - & \textgreater10hrs\\
\texttt{flat1000\usc{}50\usc{}0} &1000&245000& ? &10&106&6795.93& \textgreater10hrs & - & - & - & \textgreater10hrs\\
\texttt{flat1000\usc{}60\usc{}0} &1000&245830& ? &10&105&6368.64& \textgreater10hrs &- & - & - & \textgreater10hrs\\
\texttt{flat1000\usc{}76\usc{}0} &1000&246708&?&12&106&6628.52& \textgreater10hrs &- &-&-& \textgreater10hrs\\
\end{tabular}
\caption{\label{tab:zdd} Results from computational experiments with \texttt{B\&P+ZDD}, cells
highlighted in grey show the fastest algorithm.  Entries labeled ``init'' indicate that the
initial upper bound equaled the root lower bound, cells labels ``oom'' indicate that the algorithm
ran out of memory.  Entries of 0.00 indicate that the length of time is lower than the precision of
the timer, and a dash indicates that the information is not available.}
\end{table}

%% file: data2.tex
\begin{table}
\scriptsize
\begin{tabular}{lrr|rrrrrrr}
Instance & $n$  &  $m$  & exp&id & $Z_i$ & $Z_f$ & \% change &$|\textrm{col}|$&$t_{\textrm{price}}$\\\hline
\texttt{DSJC125.5}  &125&3891&  599&1199  &48367&52207&7.9&3462&7.9\\
\texttt{DSJC125.9}  &125&6961&  55&111  &623&627&0.6&81&0\\
\texttt{DSJC250.5}  &250&15668&  5122&10244  &1476916&1511171&2.3&37404&4108.62\\
\texttt{DSJC250.9}  &250&27897&  16411&32823  &2893&2960&2.3&808&2.88\\
\texttt{DSJC500.5}  &500&62624&  25&50  &83507135&83509619&0.003&4519&30791.1\\
\texttt{DSJC500.9}  &500&112437&  46634&93268  &15397&15913&3.4&5259&25.07\\
\texttt{DSJC1000.5}  &1000&249826&-&  -&-  & - &-&-&-\\
\texttt{DSJC1000.9}  &1000&449449&  6466&12932  &102909&105366&2.4&24892&141.52\\
\texttt{DSJR500.1c}  &500&121275&  18&37  &2443&2447&0.2&68&0.03\\
\texttt{DSJR500.5}  &500&58862&  74&149  &1809872&2222124&22.8&8&14.77\\
\texttt{le450\usc{}25c}  &450&17343&-&  -&-  & - & -&-&-\\
\texttt{le450\usc{}25d}  &450&17425&-&  -&-  & - & -&-&-\\
\texttt{queen9\usc{}9}  &81&1056&  11&23  &50719&54727&7.9&216&0.34\\
\texttt{queen10\usc{}10}  &100&2940&  76&153  &295500&308571&4.4&1088&11.99\\
\texttt{queen11\usc{}11}  &121&3960&  1902&3805  &1867378&1923347&3&16394&2138.55\\
\texttt{queen12\usc{}12}  &144&5192&  1224&2448  &12426874&12526852&0.8&24743&18255.37\\
\texttt{queen13\usc{}13}  &169&6656&  98&195  &88797420&88874820&0.09&4985&33826.38\\
\texttt{queen14\usc{}14}  &196&8372&  3&6  & 100000008  & 100000008 & 0 &5&37809.97\\
\texttt{queen15\usc{}15}  &225&10360&-&  -&-  & - & -&-&-\\
\texttt{queen16\usc{}16}  &256&12640&-&  - &- & - & -&-&-\\
\texttt{myciel3}  &11&23&  4&9  &29&29&0&0&0\\
\texttt{myciel4}  &20&71&  191&383  &152&169&11.2&0&0.02\\
\texttt{myciel5}  &47&236&  160622&321245  &1429&2188&53.1&26&35.7\\
\texttt{myciel6}  &95&755&  94395&188789&40191&70326&75&3658&357.24\\
\texttt{myciel7}  &191&2360& 7816&15632&7191878&7344883&2.1&3742&8653.59\\
\texttt{1-Insertions\usc{}4}  &67&232&  72106&144211  &85112&106550&25.2&12596&443.94\\
\texttt{1-Insertions\usc{}5}  &202&1227& - &-& - & - & -&-&-\\
\texttt{2-Insertions\usc{}4}  &149&541& - & - & - & -&-&-\\
\texttt{3-Insertions\usc{}3}  &56&110&  29128&58256&94885&183834&93.7&38580&394.57\\
\texttt{3-Insertions\usc{}4}  &281&1046&-&-& - & - & -&-&-\\
\texttt{4-Insertions\usc{}3}  &79&156& 14177&28353&6585989&6693239&1.6&12115&12196.88\\
\texttt{1-FullIns\usc{}4}  &93&593&  7422&14845  &148275&155237&4.7&1376&51.53\\
\texttt{1-FullIns\usc{}5}  &282&3247& - & - & -&- & -&-&-\\
\texttt{2-FullIns\usc{}4}  &212&1621& - & - & -&- & -&-&-\\
\texttt{2-FullIns\usc{}5}  &852&12201& - & - & -&- & -&-&-\\
\texttt{3-FullIns\usc{}4}  &405&3524& - & - & -&- & -&-&-\\
\texttt{4-FullIns\usc{}4} &690&6650& - & - & -&- & -&-&-\\
\texttt{latin\usc{}square\usc{}10}  &900&307350&  5265&10529  &52807&52807&0&13310&26.35\\
\texttt{qg.order30} &900&26100& - & - &-& - & -&-&-\\
\texttt{wap06}  &947&43571& - & - & - &-& -&-&-\\
\texttt{r250.5} &250&14849& 25&51 &137683&266893&93.8&0&0.57\\
\texttt{r1000.1c} &1000&485090& 215082&430163 &11762&11997&2&1298&165.64\\
\texttt{r1000.5} &1000&238267& 192&385 &37664084&38165484&1.3&1317&3741.64\\
\texttt{flat300\usc{}28\usc{}0} &300&21695& 1004&2009 &5817662&5858003&0.7&24665&8699.34\\
\texttt{flat1000\usc{}50\usc{}0} &1000&245000& -&- &  -  & - & -&-&-\\
\texttt{flat1000\usc{}60\usc{}0} &1000&245830& -&- &  -  & - & -&-&-\\
\texttt{flat1000\usc{}76\usc{}0} &1000&246708&-&-&  -  &-& -&-&-\\
\end{tabular}
\caption{\label{tab:zdd2} Detailed statistics for the computational experiments with
\texttt{B\&P+ZDD}.  Cells with dashes indicate that the information is not available due to time
limits or memory constraints.}
\end{table}